%% file: main.tex
\documentclass[preprint,12pt]{elsarticle}

\usepackage{color}
\usepackage{cite}
\usepackage{amsmath,amsfonts,amssymb,amsthm,bm,mathrsfs,mathtools,stmaryrd,units}
\usepackage{cleveref}
\usepackage{float}
\usepackage{graphicx}
\usepackage[table]{xcolor}
\usepackage{tikz}
\usepackage{pgfplots}
\usepackage{readarray}
\usepackage{multirow}
\usepackage[ruled,vlined]{algorithm2e}
\usepackage{enumerate}
\usepackage[nomessages]{fp}
\usepackage{subcaption}

\usepackage{xparse}
\usetikzlibrary{3d}
\usetikzlibrary{calc}
\usepgfplotslibrary{groupplots,dateplot}
\pgfplotsset{compat=1.17} 
\tikzset{xyp/.style={canvas is xy plane at z=#1}}
\tikzset{xzp/.style={canvas is xz plane at y=#1}}
\tikzset{yzp/.style={canvas is yz plane at x=#1}}

\usetikzlibrary{external}
\newtheorem{Proposition}{Proposition}

\newcommand*{\VEC}[1]  {\ensuremath{\boldsymbol{#1}}}
\newcommand*{\MAT}[1]  {\ensuremath{\boldsymbol{#1}}}

\DeclareDocumentCommand{\argmax}{O{}}{ \underset{#1}{\mathrm{argmax}} } 
\def\Hzero{\mathrm{H}_0} 
\def\Hone{\mathrm{H}_1} 
\newcommand\supinf{\underset{\Hzero}{\overset{\Hone}{\gtrless}}} 
\DeclareMathOperator{\tr}{tr}
\DeclareMathOperator{\D}{D}
\DeclareMathOperator{\grad}{grad}

\DeclareMathOperator{\logm}{logm}

\DeclarePairedDelimiter\norm{\lVert}{\rVert}
\makeatletter
\let\oldnorm\norm
\def\norm{\@ifstar{\oldnorm}{\oldnorm*}}
\makeatother

\begin{document}

\begin{frontmatter}

\title{Online Change Detection in SAR Time-Series 
with Kronecker Product Structured Scaled Gaussian Models
}

\author[label1]{Ammar Mian, Guillaume Ginolhac}
\author[label2]{Florent Bouchard}
\author[label3]{Arnaud Breloy}

\affiliation[label1]{organization={LISTIC, University Savoie Mont-Blanc},
            addressline={5 Chemin de bellevue},
            city={Annecy},
            postcode={74940},
            country={France}}

\affiliation[label2]{organization={L2S, CNRS, CentraleSupelec, University Paris-Saclay},
            addressline={3 Rue Joliot-Curie},
            city={Gif-sur-Yvette},
            postcode={91190},
            country={France}}

\affiliation[label2]{organization={LEME, University Paris Nanterre},
            addressline={50 Rue de Sevres},
            city={Ville-d'Avray},
            postcode={92410},
            country={France}}

\begin{abstract}
We develop the information geometry of scaled Gaussian distributions for which the covariance matrix exhibits a Kronecker product structure. 
This model and its geometry are then used to propose an online change detection (CD) algorithm for multivariate image times series (MITS). 
The proposed approach relies mainly on the online estimation of the structured covariance matrix under the null hypothesis, which is performed through a recursive (natural) Riemannian gradient descent.
This approach exhibits a practical interest compared to the corresponding offline version, as its computational cost remains constant for each new image added in the time series. 
Simulations show that the proposed recursive estimators reach the Intrinsic Cramér-Rao bound.
The interest of the proposed online CD approach is demonstrated on both simulated and real data. 
\end{abstract}


\end{frontmatter}

\section{Introduction}
\label{sec:intro}

\input{sections/1_introduction}

\section{A framework for robust covariance matrix change detection}
\label{sec:model}
\input{sections/2_model}

\section{Online Robust GLRT with Kronecker product structure}
\label{sec:geo}
\input{sections/3_sg_glrt_kron}

\section{Numerical experiments}
\label{sec:num_exp}
\input{sections/4_numerical_results.tex}

\section{Conclusion}

In this paper, we developed an online version of a change detection algorithm that tests the difference in a series of a structured Kronecker covariance. The data are assumed to follow a SG distribution. This model is well adapted to SAR-MITS data. The development of the online version was based on the information geometry induced by this structure and the statistical assumption. The algorithm was mainly based on the online estimation of the covariance matrix under the null hypothesis. This step was done with a Riemannian gradient and we also provided the corresponding ICRB. Results on both simulated and real data showed the good results of this approach and its practical interest to analyze large image times series. 


%
  
\bibliographystyle{elsarticle-num}
\bibliography{references}

\end{document}

%% file: sections/1_introduction.tex
Synthetic Aperture Radar (SAR) Multivariate Image Times Series (MITS) have lately been made widely available thanks to various Earth monitoring missions such as Sentinel-1, TerraSAR-X, or UAVSAR. 
This modality offers significant advantages over the optical and multispectral ones, notably when the area of interest is observed at night, or covered by clouds. 
As for optical images, the pixel of SAR images can be multivariate when exploiting the polarimetric diversity, or the spectro-angular properties of the scatterers through a wavelet decomposition \citep{MianTGRS2019}.
Hence, the analysis of multivariate SAR-MITS has become an active topic of research over the past years, with numerous applications such as change detection and crop classification.

In this paper, we focus on change detection applications for SAR-MITS~\citep{Hussain2013}.
Notably, when the local patches of the images are used instead of single pixels, the detection algorithms are mainly based on the analysis of the statistical properties of the patch over the time.
In this framework, a popular approach is based on testing the change of the covariance matrix along the time series \citep{mian2021overview}. 
An appropriate matrix distance \citep{Bazi2005,Ratha2017,Nascimento2019} is then used to take a decision. 
For time series with more than two images, a prominent class of methods builds this distance upon the Generalized Likelihood Ratio Test (GLRT), usually constructed under the assumption that the multivariate pixels follow a Gaussian distribution \citep{Conradsen2003, Ciuonzo2017}.
In practice, the radar returns makes the data distribution heavy-tailed  \citep{billingsley93,Greco2006model}.
A popular family to model non Gaussian radar data is the Complex Elliptically Symmetric (CES) distributions framework \citep{ollila2012complex}, that encompasses a significant number of well-known distribution (e.g., Weibull, Student-$t$, and $K$-distribution). 
Scaled Gaussian (SG) distributions are a major sub-family of CES, that model each sample as Gaussian conditionally to a scale factor called the texture.
The choice of the texture distribution allows then for accurately modeling the tails of the distribution.
Conversely, assuming that the texture is unknown and deterministic for each sample offers a more robust model, as it is distribution-free \citep{Tyler1987, Pascal2008a}.
This particular model has been used in \citep{Mian2019TSP} to extend the classical GLRT of \citep{Conradsen2003, Ciuonzo2017}, which improved the detection performance when using high resolution SAR images.

The performance of the aforementioned change detection methods is generally improved when the input data reflects the spectro-angular diversity of the scatterers. 
Such diversity can be accounted for through an appropriate wavelet transformation \citep{MianTGRS2019}.
This pre-processing incidentally increases the pixel dimension $p$ compared to the native polarimetric SAR pixels (for which $p=3$).
As a result, the spatial resolution of covariance-based change detectors is reduced, as the tested local patches should include a number of pixel $n$ that scales with $p$.
This issue can be mitigated by assuming an additional structure for the covariance matrix and derive the corresponding change detection algorithm.
This approach was successfully leveraged by considering low-rank structures of the covariance matrix in \citep{abdallah2019detection, MianJSTARS2020}. 
A limitation regarding these methods is that the rank of the covariance matrix is an extra parameter that needs to be estimated.
Moreover, this rank often changes with respect to the studied pixel and the images which leads to complicated strategies for the change detection process. 
This motivates the use of alternate low-dimensional covariance matrix structures.

In this work, we remark that the transformation proposed in \citep{MianTGRS2019} combined with polarization channels naturally leads to have an inherent Kronecker product structure.
This structure is interesting since it does not involve any nuisance parameter: the dimensions of the sub-matrices are fixed by the number of polarization channels.
It also greatly reduces the number of parameters to be estimated, which allows for preserving a thin spatial resolution. 
The Kronecker product structure has been extensively leveraged in the radar community and for MIMO systems.
For example, several algorithms have been developed for the estimation of Kronecker product structured covariance matrices when the distribution is assumed to be Gaussian \citep{werner2008estimation,srivastava2008models}. 
Their extensions to Scaled Gaussian (SG) model, have been developed in \citep{wiesel2012geodesic, sun2016robust}. 
To the best of our knowledge, this structure was not considered in change detection for SAR-MITS.
Hence, building upon the aforementioned works, the first main contribution of this paper is thus to derive the GLRT for the change detection problem when the data follows a SG distribution with a Kronecker structure for the covariance matrix.

We then address important issues regarding the computational load of robust change detection methods.
Indeed, the computational cost of robust change detection \citep{Mian2019TSP, MianJSTARS2020} becomes prohibitive when the number of images $T$ increases, which limits their practical use for large SAR-MITS. 
The main computational bottleneck is the computation of a robust covariance matrix estimate, that needs to be re-evaluated under the null hypothesis for each new image.
This issue can be mitigated by replacing the estimate by one computed online (i.e., in a one pass streaming process).
In this scope, the Riemannian optimization framework \citep{absil2009optimization} offers an interesting lead since the covariance matrix is a parameter that belongs to a smooth manifold: the manifold of Hermitian Positive Definite (HPD) matrices.
Especially, a Riemannian gradient based online procedure proposed in \citep{zhou2019fast} allows for obtaining a statistically efficient estimate in a computationally efficient manner (as the gradient steps do not require any grid search).
This method has already been successfully leveraged for scaled Gaussian models in \citep{BMZSGB20} and {\it t}-distributed data in \citep{Bouchard2021}.
The second main contribution of this paper is to leverage this framework for reducing the computational load of the proposed change detection method.
This requires to study the information geometry of scaled Gaussian models with Kronecker product structured covariance matrices.
As a side result of this analysis, we also obtain the intrinsic Cramér-Rao bound~\citep{Smith2005, breloy2019intrinsic} for the corresponding estimation problem.
This bound is then used to validate the proposed online estimation method on simulated data.

Finally, we demonstrate the interest of the proposed approach on a real dataset provided by UAVSAR (courtesy of NASA/JPL-Caltech).
The SAR-MITS is referenced as SDelta\_28518\_02, Segment 1, which shows the evolution of a river delta in the USA with cycle of droughts and flood. 

The rest of the article is organized as follows: 
Section \ref{sec:model} presents the data model and the GLRT adapted to a Kronecker product structured covariance matrix. 
Section \ref{sec:geo} studies the information geometry of this model, and presents the derivation of the corresponding online estimation algorithm. 
Section \ref{sec:num_exp} illustrates the performance of the proposed method on simulated and real data.
Notations: italic indicates a scalar quantity, lower case boldface indicates a vector quantity, and upper case boldface a matrix. The transpose conjugate operator is $^H$ and the conjugate one is $^*$. $tr(\cdot)$ and $|\cdot|$ are respectively the trace and the determinant operators. 
$\mathcal{H}^{++}_{p}$ is the manifold of HPD matrices of size $p \times p$.  
For $\mathbf{x}\in \mathbb{C}^p$, the notation $\mathbf{x} \sim \mathcal{CN} (\boldsymbol{\nu}, \boldsymbol{\Sigma})$ stands for a complex-valued random Gaussian vector of mean $\boldsymbol{\nu} \in \mathbb{C}^p $ and covariance matrix $\boldsymbol{\Sigma}\in \mathcal{H}^{++}_{p}$.

%% file: sections/2_model.tex

\subsection{Setup}

Within a SAR-MITS, we consider detecting a change within a local (spatial) patch of pixel.
Given a set of $T$ co-registered images, this means processing $T$ sets of $n$ pixels of dimension $p$, denoted $\{ \{\VEC{x}_i^{(t)}\}_{i\in\llbracket1,n\rrbracket} \}_{t\in\llbracket1,T\rrbracket} $.
The corresponding setup is displayed in Figure \ref{fig:Data description}. 
A change detection procedure can then be derived by: $i$) expressing a relevant statistical model and corresponding likelihood function for the dataset; $ii$) deriving a change detection test (such as the GLRT) to assess whether the parameters of the statistical model change over time, or not.
In the context of SAR-MITS, such procedures usually test for a change within the covariance matrix of the patch, as it is a meaningful feature for this type of data \citep{mian2021overview, Bazi2005, Ratha2017, Nascimento2019, Conradsen2003, Ciuonzo2017, Mian2019TSP}.

\input{figures/setup}

\subsection{Robust statistical model for the pixel patches}

\label{sec:scaled_gaussian}

For a given time $t$ the pixel patch consists of $n$ i.i.d. samples $\{\VEC{x}_i^{(t)}\}_{i\in\llbracket1,n\rrbracket}$ where each sample $\VEC{x}_i^{(t)}\in\mathbb{C}^{p}$. 
Most work modeled this data as following a complex Gaussian distribution \citep{Conradsen2003, Ciuonzo2017}, i.e., that pixels are drawn from $\VEC{x}_i^{(t)}\sim \mathcal{CN}(\mathbf{0},\MAT{\Sigma}^{(t)} ) \,\, \forall i$, with $\MAT{\Sigma}^{(t)} \in \mathcal{H}_p^{++} $.
However, it was shown in \citep{Mian2019TSP} that the complex SG distribution 
offers a better fit to high resolution SAR images, and allows for improving the performance of change detection.
In this case each sample $\VEC{x}_i^{(t)}$ is modeled as complex Gaussian conditionally to an unknown scale, referred to as the texture.
The choice of the texture distribution allows then for modeling many well known heavy-tailed distributions \citep{ollila2012complex}.
Following from the robust estimation framework \citep{Tyler1987, Pascal2008a}, a distribution-free detection process will be obtained by assuming that the texture is unknown and deterministic for each sample.
Hence, we model samples as $\VEC{x}_i^{(t)} \sim \mathcal{CN}(\VEC{0}, \tau_i^{(t)}\boldsymbol{\Sigma}^{(t)})$ with $\tau_i^{(t)}\in \mathbb{R}_*^{+}, \forall~i\in [\![1,n]\!]$. 
The log-likelihood of the pixel patch at time $t$ is then, up to a constant,
\begin{equation}
\mathcal{L}(\{\VEC{x}_i^{(t)}\}_{i\in\llbracket1,n\rrbracket},\theta^{(t)}) 
=\sum_{i=1}^n \log|\tau_i^{(t)} \MAT{\Sigma}^{(t)}| + \frac{1}{\tau_i^{(t)}}\left.\VEC{x}_i^{(t)}\right.^H \left.\MAT{\Sigma}^{(t)}\right.^{-1}\VEC{x}_i
\label{eq:logpdf}
\end{equation}
where $\theta^{(t)}=\{\MAT{\Sigma}^{(t)}, \boldsymbol{\tau}^{(t)}\}$ is the set of unknown statistical parameters, in which $\boldsymbol{\tau}^{(t)} = [ \tau_1^{(t)},~\cdots,~\tau_n^{(t)}  ]$ denotes the vector of textures.
Though this will not be exploited in this work, we also notice that the Gaussian model can be recovered as a sub-case by simply setting $\tau_i^{(t)}=1, \forall~i\in [\![1,n]\!]$.


\subsection{Change detection with the GLRT}

Covariance based change detection operates by assuming that a change in the image is translated by a change in the unknown parameters $\theta^{(t)}$ over the time index $t$.
From the SAR-MITS $\{ \{\VEC{x}_i^{(t)}\}_{i\in\llbracket1,n\rrbracket} \}_{t\in\llbracket1,T\rrbracket} $, the change detection problem can thus be written as the following hypothesis test:
\begin{equation}
    \left\{
        \begin{array}{ll}
          \mathrm{H_0}: \theta^{(1)}=\theta^{(2)}=...=\theta^{(T)} \overset{\Delta}{=} \theta^{(0)} \vspace{0.1cm}\\
          \mathrm{H_1}: \exists (t,t')\in\llbracket 1,T\rrbracket^2,~ \theta^{(t)} \ne \theta^{(t^\prime)}\\
        \end{array}
    \right.
    \label{eq:ProblemDetection}
\end{equation}
For the SG model of Section \ref{sec:scaled_gaussian}, the GLRT corresponding to this hypothesis test has been derived in \citep{Mian2019TSP} and is expressed as:
\begin{equation}
	\hat{\Lambda}_{SG}  = \frac{\left|\hat{\mathbf{\Sigma}}^{(0)}\right|^{Tn}}{\displaystyle \prod_{t=1}^T \left| {\hat{\mathbf{\Sigma}}^{(t)}}\right|^n} \displaystyle\prod_{\substack{i =1}}^{\substack{n}}  \frac{ \left(\displaystyle   \hat{\tau}_{i}^{(0)} \right)^{Tp} }{\displaystyle \prod_{t=1}^T\left( \hat{\tau}_i^{(t)} \right)^{p}} \supinf \lambda ,
\label{eq:GLRTMatGen}
\end{equation}
where $\hat{\theta}^{(0)}=\{\hat{\MAT{\Sigma}}^{(0)}, \hat{\boldsymbol{\tau}}^{(0)}\}$
stands for the maximum likelihood estimator (MLE) of the parameters under $\mathrm{H_0}$, and 
$\hat{\theta}^{(t)}=\{\hat{\MAT{\Sigma}}^{(t)}, \hat{\boldsymbol{\tau}}^{(t)}\},~\forall t\in [\![1,T]\!]$ for the MLE of the parameters under $\mathrm{H_1}$.
These MLEs are expressed respectively as
\begin{eqnarray}
	\hat{\mathbf{\Sigma}}^{(0)}  = \frac{p}{n} \sum_{i =1}^{n}  \frac{\displaystyle \sum_{t=1}^T \VEC{x}_i^{(t)} \left.\VEC{x}_i^{(t)}\right.^H}{\displaystyle \sum_{t=1}^T \left.\VEC{x}_i^{(t)}\right.^H \left.\hat{\mathbf{\Sigma}}^{(0)}\right.^{-1} \VEC{x}_i^{(t)}}
	~~\text{and}~~
	\hat{\tau}_{i}^{(0)}=\sum_{t=1}^T\frac{\left.\VEC{x}_i^{(t)}\right.^H \left.\hat{\mathbf{\Sigma}}^{(0)}\right.^{-1}\VEC{x}_i^{(t)}}{Tp},
	\label{eq:defTylerNullHypo}
\end{eqnarray}
and
\begin{eqnarray}
	\hat{\mathbf{\Sigma}}^{(t)}  = \frac{p}{n} \sum_{i =1}^{n}  \frac{\displaystyle \VEC{x}_i^{(t)} \left.\VEC{x}_i^{(t)}\right.^H}{\displaystyle \left.\VEC{x}_i^{(t)}\right.^H \left.\hat{\mathbf{\Sigma}}^{(t)}\right.^{-1} \VEC{x}_i^{(t)}}
	~~\text{and}~~
	\hat{\tau}_{i}^{(t)}= \frac{\left.\VEC{x}_i^{(t)}\right.^H \left.\hat{\mathbf{\Sigma}}^{(t)}\right.^{-1}\VEC{x}_i^{(t)}}{p},
	\label{eq:defTylerH1Hypo}
\end{eqnarray}
with the latter corresponds to Tyler's $M$-estimator of the scatter matrix~\citep{Tyler1987,Pascal2008a}.


\subsection{Some practical limitations}

\noindent
\textbf{Data dimension $p$} (pixel depth):
For SAR-MITS, the diversity in the data arises from polarization channels (HH, VV and HV/VH), and the spectro-angular diversity of the scatterers which can be obtained by using an appropriate wavelet transform \citep{MianTGRS2019}.
The latter transformation tends to improve the detection performance of the covariance based change detectors (such as $\hat{\Lambda}_{SG}$ in \eqref{eq:GLRTMatGen}), however it increases the pixel depth $p$.
It is consequently detrimental to the spatial resolution of the detectors, as the number of pixel in the patch $n$ needs to scale with $p$ to accurately estimate the local covariance matrix.
The assumption of i.i.d. data can also be violated when processing large areas, which motivates the development of methods that reduce the number of needed samples.
This reduction can be achieved by using regularization \citep{ollila2014regularized}, or low-rank structured covariance matrices \citep{MianJSTARS2020}.
Nevertheless these methods involve regularization parameters that need proper tuning, which increases the complexity of the detection process.
This issues will be addressed in Section \ref{sec:glrt_offline}, where we leverage a Kronecker product structure whose hyper-parameters are inherently set by the wavelet transform \citep{MianTGRS2019}.\\

\noindent
\textbf{Data dimension $T$} (length of the time series):
The detector $\hat{\Lambda}_{SG}$ in \eqref{eq:GLRTMatGen} involves solution of fixed point equations \eqref{eq:defTylerH1Hypo} and \eqref{eq:defTylerNullHypo}.
Especially, the fixed point in \eqref{eq:defTylerNullHypo} depends on the whole dataset and cannot be computed recursively (i.e., in a streaming fashion).
Whether a low-dimensional structure is used or not, this means that the computational load of the robust detectors increases heavily for each new added image in the stack, which can be limiting.
This issue is addressed in Section \ref{sec:geo}, where we leverage the information geometry of the considered statistical model in order to propose a statistically efficient on-line estimation process.

\section{Robust GLRT with Kronecker product structure}
\label{sec:glrt_offline}

The wavelet transform in \citep{MianTGRS2019} consists in combining the polarization channels to reflect the spectro-angular diversity of the data. 
This combination implies a Kronecker product structure of the covariance matrix, in which the size of one of the two matrices will be $3$ (number of polarization channels), while the dimension of the other will be equal to the number of frequency times the angular intervals that are considered.
Hence, the covariance matrix $\MAT{\Sigma}$ can be assumed to be structured as $\MAT{\Sigma} = \MAT{A} \otimes \MAT{B}$, with $ \MAT{A}\in \mathcal{H}_a^{++}$ and $\MAT{B}\in \mathcal{H}_b^{++} $, where $a=3$ and $b = p/3$\footnote{This remains valid whatever $a$ and $b$.}.
To avoid any scaling ambiguity on this decomposition and the scaling by the textures, an arbitrary normalization on $\MAT{A}$ and $\MAT{B}$ is needed. 
We rely on the unit determinant for both matrices $\MAT{A}$ and $\MAT{B}$, which, as noted in~\citep{Paindaveine2008,BMZSGB20}, is particularly advantageous from a geometrical point of view.
In this case, we have $\MAT{A}\in s\mathcal{H}^{++}_{a}=\{\MAT{M}\in\mathcal{H}^{++}_{a}:|\MAT{M}|=1\}$ and $\MAT{B}\in s\mathcal{H}^{++}_{b}=\{\MAT{M}\in\mathcal{H}^{++}_{b}:|\MAT{M}|=1\}$.
Following from the SG model of Section \ref{sec:model}, we have that each sample of a patch is distributed according to $\VEC{x}_i \sim \mathcal{CN}(\mathbf{0}, \tau_i\MAT{A}\otimes\MAT{B}$).
The likelihood of a pixel patch at time $t$ is then directly transposed from \eqref{eq:logpdf}, and also denoted 
$\mathcal{L}(\{\VEC{x}_i^{(t)}\}_{i\in\llbracket1,n\rrbracket},\theta^{(t)}) $, in which the parameter of interest $\theta^{(t)}$ now denotes $\theta^{(t)}=\{ \MAT{A}^{(t)},\MAT{B}^{(t)},\boldsymbol{\tau}^{(t)} \}$, which lies in the product manifold $\mathcal{M}=s\mathcal{H}^{++}_a\times s\mathcal{H}^{++}_b\times\mathbb{R}^{n}_{++}$.
The following proposition then adapts the GLRT of \eqref{eq:GLRTMatGen} to the case where the covariance matrix has such Kronecker product structure.

\begin{Proposition}
The GLRT for the problem of detection \eqref{eq:ProblemDetection} when $\theta^{(t)}=\{\MAT{A}^{(t)},\MAT{B}^{(t)},\VEC{\tau}^{(t)}\}$ and $T$ images is:\footnotesize
\begin{equation}
    \log ( \hat{\Lambda}_{K-SG}^{(T)} ) = 
    \mathcal{L}^{\Hone} ( \{ \hat{\theta}^{(t)} \}_{t=1}^T ) 
    -
    \mathcal{L}^{\Hzero} ( \hat{\theta}^{(0)} )
    %
    \label{eq:GLRTKronCG}
\end{equation}\normalsize
where the log-likelihoods of the whole data under $\Hzero$ and $\Hone$ are expressed as:
\begin{equation}
    \begin{array}{rll}
         \mathcal{L}^{\Hzero}  ( {\theta}^{(0)} ) & = & -n T \log|\MAT{A}^{(0)}\otimes\MAT{B}^{(0)}|-Tp\sum_{i=1}^n \log(\tau_i^{(0)}) \vspace{0.2cm}\\
         & &-\sum_{i=1}^n \sum_{t=1}^T\frac{\left.\VEC{x}_i^{(t)}\right.^H(\MAT{A}^{(0)}\otimes\MAT{B}^{(0)})^{-1}\VEC{x}_i^{(t)}}{\tau_{i}^{(0)}}  \vspace{0.4cm} \\
         \mathcal{L}^{\Hone} ( \{ {\theta}^{(t)} \}_{t=1}^T )  & = & -n\sum_{t=1}^T \log|\MAT{A}^{(t)}\otimes\MAT{B}^{(t)}|-p\sum_{i=1}^n \sum_{t=1}^T\log(\tau_i^{(t)}) \vspace{0.2cm}\\
         & &-\sum_{i=1}^n \sum_{t=1}^T\frac{\left.\VEC{x}_i^{(t)}\right.^H(\MAT{A}^{(t)}\otimes\MAT{B}^{(t)})^{-1}\VEC{x}_i^{(t)}}{\tau_i^{(t)}},
    \end{array}
    \label{eq:logH0H1}
\end{equation}
and where $\hat{\theta}^{(0)}=\{\hat{\MAT{A}}^{(0)}, \hat{\MAT{B}}^{(0)}, \hat{\boldsymbol{\tau}}^{(0)}\}$
stands for the maximum likelihood estimator (MLE) of the parameters under $\mathrm{H_0}$, and 
$\hat{\theta}^{(t)}=\{\hat{\MAT{A}}^{(t)}, \hat{\MAT{B}}^{(t)}, \hat{\boldsymbol{\tau}}^{(t)}\},~\forall t\in [\![1,T]\!]$ for the MLE of the parameters under $\mathrm{H_1}$.
These MLEs are expressed respectively as
\begin{equation} \footnotesize
    \begin{array}{l}
         \hat{\MAT{A}}^{(0)} = \frac{1}{T n b} \left( \sum_{t=1}^T \sum_{i=1}^n \frac{\MAT{M}_{i,t}^T(\hat{\MAT{B}}^{(0)})^{-1*}\MAT{M}_{i,t}^{*}}{\hat{\tau}_{i}^{(0)}}\right) \vspace{0.2cm} \\
         \hat{\MAT{B}}^{(0)} = \frac{1}{T n a} \left( \sum_{t=1}^T \sum_{i=1}^n \frac{\MAT{M}_{i,t}(\hat{\MAT{A}}^{(0)})^{-1*}\MAT{M}_{i,t}^{H}}{\hat{\tau}_{i}^{(0)}}\right) 
    \end{array}
    ~\textbf{and}~
    \hat{\tau}_{i}^{(0)} =\sum_{t=1}^T \frac{\left.\VEC{x}_i^{(t)}\right.^H(\hat{\MAT{A}}^{(0)}\otimes\hat{\MAT{B}}^{(0)})^{-1}\VEC{x}_i^{(t)}}{T p}
    \label{eq:MLE_H0}
\end{equation}
and
\begin{equation} \footnotesize
    \begin{array}{l}
         \hat{\MAT{A}}^{(t)} = \frac{1}{n b} \left( \sum_{i=1}^n \frac{\MAT{M}_{i,t}^T(\hat{\MAT{B}}^{(t)})^{-1*}\MAT{M}_{i,t}^{*}}{\hat{\tau}_{i}^{(t)}}\right) \vspace{0.2cm} \\
         \hat{\MAT{B}}^{(t)} = \frac{1}{n a} \left( \sum_{i=1}^n \frac{\MAT{M}_{i,t}(\hat{\MAT{A}}^{(t)})^{-1*}\MAT{M}_{i,t}^{H}}{\hat{\tau}_{i}^{(t)}}\right) 
    \end{array}
    ~\textbf{and}~
    \hat{\tau}_{i}^{(t)} = \frac{\left.\VEC{x}_i^{(t)}\right.^H(\hat{\MAT{A}}^{(t)}\otimes\hat{\MAT{B}}^{(t)})^{-1}\VEC{x}_i^{(t)}}{p}
    \label{eq:MLE_H1}
\end{equation}
where $\MAT{M}_{i,t}\in \mathbb{C}^{a \times b}$ is matrix 
storing the entries of $\VEC{x}_i^{(t)}$, such that ${\rm vec}(\MAT{M}_{i,t}) =\VEC{x}_i^{(t)}$. All these MLE estimators in \eqref{eq:MLE_H0} and \eqref{eq:MLE_H1} can be evaluated with a fixed-point algorithm, such as the one proposed in \citep{sun2016robust}.
\end{Proposition}
\begin{proof}
The principle of the GLRT consists in computing the quantity: 
\begin{equation}
\log (\hat{\Lambda}_{K-SG}^{(T)})  = 
\underset{\{ {\theta}^{(t)} \}_{t=1}^T }{\mathrm{max}} \quad
\mathcal{L}^{\Hone} ( \{ {\theta}^{(t)} \}_{t=1}^T ) 
-
\underset{  {\theta}^{(0)}   }{\mathrm{max}} \quad
\mathcal{L}^{\Hzero} (  {\theta}^{(t)}   ) .
%
\end{equation}
Obtaining the maximal values in the above expression then boils down to computing the MLE under $\Hzero$ and $\Hone$.
This requires canceling the derivative of $\mathcal{L}^{\Hzero} $
(resp. $\mathcal{L}^{\Hone}  $)
with respect to $\MAT{A}^{(0)}$, $\MAT{B}^{(0)}$, and $\VEC{\tau}^{(0)}$ (resp. $\MAT{A}^{(t)}$, $\MAT{B}^{(t)}$, and $\VEC{\tau}^{(t)}$ $\forall t$).
The derivatives of the trace terms are obtained by using $\VEC{x}^H (\MAT{A}\otimes \MAT{B})^{-1} \VEC{x} = \tr(\MAT{A}^{-T}\MAT{M}^H\MAT{B}^{-1}\MAT{M})$ and the following derivative:
\begin{equation}
    \frac{\partial \tr(\MAT{O}\MAT{X}^{-1}\MAT{P})}{\partial \MAT{X}} = -(\MAT{X}^{-1}\MAT{P}\MAT{O}\MAT{X}^{-1})^H ,
\end{equation}
which yields the following results:
\begin{equation}
    \begin{array}{lll}
      \frac{\partial \VEC{x}^H (\MAT{A}\otimes \MAT{B})^{-1} \VEC{x}}{\partial \MAT{A}}   &  = & -\MAT{A}^{-1}\MAT{M}^T(\MAT{B}^{-1})^{*}\MAT{M}^{*}\MAT{A}^{-1} \vspace{0.2cm}\\
       \frac{\partial \VEC{x}^H (\MAT{A}\otimes \MAT{B})^{-1} \VEC{x}}{\partial \MAT{B}}   &  = & -\MAT{B}^{-1}\MAT{M}(\MAT{A}^{-1})^{*}\MAT{M}^{H}\MAT{B}^{-1} .
    \end{array}
\end{equation}
The derivatives of the log terms are obtained by using the relation $\log |\MAT{A}\otimes \MAT{B}| = b\log |\MAT{A}|+a\log |\MAT{B}|$, and $\frac{\partial \log |\MAT{X}|}{\partial \MAT{X}}=\MAT{X}^{-H}$.
After some matrix manipulations, we obtain the final results in \eqref{eq:MLE_H0} and \eqref{eq:MLE_H1}. 
%
\end{proof}

\noindent
\textbf{Remark}: the expression of the MLEs under $\Hone$ coincide with Tyler's $M$-estimator with Kronecker product structure obtained in \citep{sun2016robust} (slightly modified, as we deal with complex valued matrices).
Those under $\Hzero$ are slightly different due to the assumed shared parameters between the different time index, and thus involve a summation over $t$.\\

Similarly to the unstructured GLRT in \eqref{eq:GLRTMatGen}, the detector \eqref{eq:GLRTKronCG} cannot be recursively derived when a new patch $\mathcal{X}_{T+1}=\{\VEC{x}_i^{(T+1)}\}_{i=1}^n$ is added to the image stack. 
Indeed, if we denote the observation sets $\mathcal{X}_{1:T} =  \{ \{\VEC{x}_i^{(t)}\}_{i=1}^n \}_{t=1}^{T} $ and $\mathcal{X}_{1:T+1} =  \{ \{\VEC{x}_i^{(T+1} \}_{i=1}^n \}_{t=1}^{T+1} $, it is easy to notice for the MLE of $\MAT{A}$ under $\Hzero$ that: 
\begin{equation}
	\hat{\MAT{A}}^{(0)} ( \mathcal{X}_{1:T+1}    )  \neq \frac{T}{T+1} \hat{\MAT{A}}^{(0)} ( \mathcal{X}_{1:T}   ) + \frac{1}{T+1}\hat{\MAT{A}}^{(0)} (  \mathcal{X}_{T+1}    )
\end{equation}
due to the intricacy of the fixed-point solutions (the same is also true for $\MAT{B}$ and so on $\VEC{\tau}$).
When processing large time series, re-computing the fixed point can then be computationally demanding for large dimensions $p$ and $T$, which limits the practical implementation of the GLRTs.
To overcome this issue, we propose in the next section to recursively update the parameters $\theta^{(0)}$, which allows for computing an online version of the detector \eqref{eq:GLRTKronCG}. 
Such recursive algorithm, inspired by \citep{BMZSGB20}, will be obtained by leveraging the information geometry of the manifold $\mathcal{M}$ induced by the SG distribution.

%% file: figures/setup.tex
\begin{figure}[!t]
	\centering
	\newcommand\nmax{5} 
	\newcommand\mmax{5} 
	\newcommand{\omax}{3}
	\resizebox{\columnwidth}{!}{%
	\begin{tikzpicture}[scale=0.7]
		\fill[xyp=\omax, gray!30] (1,1) rectangle (4,4); 
		\foreach \n in {1,...,\nmax}
		{   
			\foreach \o in {1,...,\omax}
			{
				\draw[xzp=\mmax] (\n-1,0) rectangle (\n,\o);
			}
			\foreach \m in {1,...,\mmax}
			{
				\draw[xyp=\omax] (\n-1,\m-1) rectangle (\n,\m);
				
			}
		}  
		\foreach \m in {1,...,\mmax}
		{
			\foreach \o in {1,...,\omax}
			{	
				\draw[yzp=\nmax] (\m-1,0) rectangle (\m,\o);
			}
		}
		\draw[xzp=\mmax] [<->] (-0.25,0) -- (-0.25,\omax);
		\draw[xzp=\mmax] (-0.6,1.5) node {$p$};
		\draw[xyp=\omax] (2.5,-0.5) node {$t=1$};
		\foreach \m in {1,2,3}
		{ 
			\foreach \n in {1,2,3}
			{ 
				\FPeval{\result}{clip(\n+(\m-1)*3)}
				\draw[xyp=\omax] (\n,\m) ++ (0.5,0.5) node {\small $\VEC{x}_\result^{(1)}$};
			}
		}
		\draw[xyp=\omax] (7,2.5) node {\ldots};
		\fill[xyp=\omax, gray!30] (9,1) rectangle (12,4); 
		\foreach \n in {1,...,\nmax}
		{   
			\FPeval{\nd}{clip(\n+8)}
			\foreach \o in {1,...,\omax}
			{
				\draw[xzp=\mmax] (\nd-1,0) rectangle (\nd,\o)  ;
			}
			\foreach \m in {1,...,\mmax}
			{
				\draw[xyp=\omax] (\nd-1,\m-1)  rectangle (\nd,\m)  ;
				
			}
		}  
		\foreach \m in {1,...,\mmax}
		{
			\foreach \o in {1,...,\omax}
			{	
				\FPeval{\result}{clip(\nmax+8)}
				\draw[yzp=\result] (\m-1,0)   rectangle (\m,\o);
			}
		}
		\draw[xzp=\mmax] [<->] (7.75,0) -- (7.75,\omax);
		\draw[xzp=\mmax] (7.4,1.5) node {$p$};
		\draw[xyp=\omax] (10.5,-0.5) node {$t=T$};
		
		\foreach \m in {1,2,3}
		{ 
			\foreach \n in {1,2,3}
			{ 
				\FPeval{\result}{clip(\n+(\m-1)*3)}
				\draw[xyp=\omax] (\n,\m) ++ (8.5,0.5) node {\small $\VEC{x}_\result^{(T)}$};
			}
		}
		\end{tikzpicture}
		}
		\vspace{-1em}
    \caption{Illustration data notation for change detection test in a MITS.}
    \label{fig:Data description}
	\end{figure}
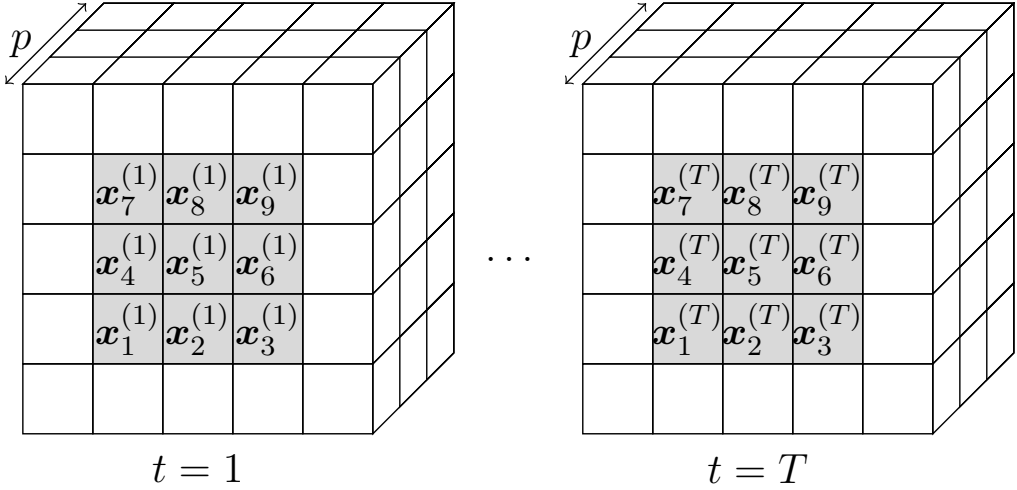	

%% file: sections/3_sg_glrt_kron.tex
In this section, we first derive the Riemannian geometry of $\mathcal{M}$ equipped with the Fisher information metric induced by the likelihood~\eqref{eq:logpdf}. 
The exponential mapping which is needed in any Riemannian gradient descent algorithm is also given as well as the corresponding geodesic distance. 
The second part of this section is devoted to the online estimation of the parameters $\theta^{(0)}$ with a stochastic Riemannian gradient descent algorithm.

\subsection{Information geometry of $\mathcal{M}$} 

In this subsection, we will omit the subscript $.^{(t)}$ since the results given in this section will be valid for any values of $T$. By definition, the tangent space of $\mathcal{M}$ at $\theta$ is $T_{\theta} \mathcal{M} = T_{\MAT{A}}s\mathcal{H}^{++}_{a} \times T_{\MAT{B}}s\mathcal{H}^{++}_{b} \times T_{\VEC{\tau}}\mathbb{R}_{++}^n$.
The tangent space $T_{\VEC{\tau}}\mathbb{R}_{++}^n$ is identified to $\mathbb{R}^n$ and the tangent spaces $T_{\MAT{A}}\mathcal{H}^{++}_{a}$ and $T_{\MAT{B}}\mathcal{H}^{++}_{b}$ can be identified by \citep{breloy2019intrinsic}
\begin{equation}
\begin{array}{lll}
     T_{\MAT{A}}s\mathcal{H}^{++}_a & = &  \{ \MAT{\xi}_{\MAT{A}} \in \mathcal{H}_a: \, \tr(\MAT{A}^{-1}\MAT{\xi}_{\MAT{A}}) = 0\}.\\
     T_{\MAT{B}}s\mathcal{H}^{++}_b & = &  \{ \MAT{\xi}_{\MAT{B}} \in \mathcal{H}_b: \, \tr(\MAT{B}^{-1}\MAT{\xi}_{\MAT{B}}) = 0\}.\\
\end{array}
	\label{eq:TangentSpace}
\end{equation} 
In the following, $\xi = (\MAT{\xi}_{\MAT{A}}, \MAT{\xi}_{\MAT{B}},\MAT{\xi}_{\VEC{\tau}})$ and $\eta = (\MAT{\eta}_{\MAT{A}}, \MAT{\eta}_{\MAT{B}},\MAT{\eta}_{\VEC{\tau}})$ denote two elements from $T_{\theta} \mathcal{M}$.
The Fisher information metric on $\mathcal{M}$ induced by the scaled Gaussian distribution is provided in Proposition~\ref{prop:fim}.

\begin{Proposition}[Fisher information metric]
	\label{prop:fim}
	Given $\theta\in\mathcal{M}$, $\xi$ and $\eta\in T_{\theta}\mathcal{M}$, the Fisher information metric on $\mathcal{M}$ induced by the likelihood~\eqref{eq:logpdf} is
	\begin{multline*}
		\langle \xi, \eta\rangle_{\theta} =
		\frac{b}{p} \tr(\MAT{A}^{-1}\MAT{\xi}_{\MAT{A}}\MAT{A}^{-1}\MAT{\eta}_{\MAT{A}})
		+ \frac{a}{p} \tr(\MAT{B}^{-1}\MAT{\xi}_{\MAT{B}}\MAT{B}^{-1}\MAT{\eta}_{\MAT{B}})
		\\
		+ \frac{1}{n}(\MAT{\xi}_{\VEC{\tau}}\odot \VEC{\tau}^{\odot^{-1}})^T(\MAT{\eta}_{\VEC{\tau}}\odot \VEC{\tau}^{\odot^{-1}})
	\end{multline*}
\end{Proposition}
\begin{proof}
    We first introduce the mapping $\varphi_i:\mathcal{M}\to\mathcal{H}^{++}_p$ as:
    \begin{equation}
        \varphi_i(\theta^{(t)}) = \tau_i^{(t)}\MAT{A}^{(t)}\otimes\MAT{B}^{(t)},
    \end{equation}
    such that $\VEC{x}_i^{(t)} \sim \mathcal{CN}(\mathbf{0}, \varphi_i(\theta^{(t)}))$.
	From~\citep[Proposition~7]{bouchard2020riemannian}, we have
	\begin{equation*}
		\langle\xi,\eta\rangle_{\theta} = \langle\D\varphi_i(\theta)[\xi],\D\varphi(\theta)[\eta]\rangle_{\varphi(\theta)}^{\mathcal{H}^{++}_p},
	\end{equation*}
	where $\langle\cdot,\cdot\rangle^{\mathcal{H}^{++}_p}_{\cdot}$ is the Fisher metric of the Gaussian distribution on $\mathcal{H}^{++}_p$.
    This metric is, e.g., obtained in \citep{Smith2005, breloy2019intrinsic}, and is defined by:
	\begin{equation}
	    \langle\MAT{\Sigma}_1,\MAT{\Sigma}_2\rangle^{\mathcal{H}^{++}_p}_{\MAT{\Sigma}} = \tr(\MAT{\Sigma}^{-1}\MAT{\Sigma}_1\MAT{\Sigma}^{-1}\MAT{\Sigma}_2)
	    \label{eq:fimgaussian}
	\end{equation}
	The directional derivative of $\varphi$ at $\theta$ in the direction $\xi$ is obtained by $\phi_i(\theta+\xi)-\phi_i(\theta)$ and the result is then
	\begin{equation*}
		\D\varphi(\theta)[\xi] = \tau_i \MAT{A}\otimes\MAT{\xi}_{\MAT{B}} + \tau_i \MAT{\xi}_{\MAT{A}}\otimes\MAT{B} + \MAT{\xi}_{\tau_i} \MAT{A}\otimes \MAT{B}.
	\end{equation*}
	The result is obtained by plugging $\D\varphi(\theta)[\xi]$ into $\langle\cdot,\cdot\rangle^{\mathcal{H}^{++}_p}_{\MAT{\Sigma}}$, we have:
	\begin{multline*}
	        \langle\xi,\eta\rangle_{\theta}  =  \sum_{i=1}^n  \langle \tau_i \MAT{A}\otimes\MAT{\xi}_{\MAT{B}} + \tau_i \MAT{\xi}_{\MAT{A}}\otimes\MAT{B} + \MAT{\xi}_{\tau_i} \MAT{A}\otimes \MAT{B}, \\
	        \tau_i \MAT{A}\otimes\MAT{\eta}_{\MAT{B}} + \tau_i \MAT{\eta}_{\MAT{A}}\otimes\MAT{B} + \MAT{\eta}_{\tau_i} \MAT{A}\otimes \MAT{B}\rangle^{\mathcal{H}^{++}_p}_{\MAT{\Sigma}} \\
	        = \tr\left\{ (\tau_i\MAT{A}\otimes\MAT{B})^{-1}(\tau_i \MAT{A}\otimes\MAT{\xi}_{\MAT{B}} + \tau_i \MAT{\xi}_{\MAT{A}}\otimes\MAT{B} + \MAT{\xi}_{\tau_i} \MAT{A}\otimes \MAT{B})\right. \\
	       \left. (\tau_i\MAT{A}\otimes\MAT{B})^{-1}(\tau_i \MAT{A}\otimes\MAT{\eta}_{\MAT{B}} + \tau_i \MAT{\eta}_{\MAT{A}}\otimes\MAT{B} + \MAT{\eta}_{\tau_i} \MAT{A}\otimes \MAT{B})\right\}
	\end{multline*}
	Let us compute
	\begin{multline*}
	(\tau_i\MAT{A}\otimes\MAT{B})^{-1}(\tau_i \MAT{A}\otimes\MAT{\xi}_{\MAT{B}} + \tau_i \MAT{\xi}_{\MAT{A}}\otimes\MAT{B} + \MAT{\xi}_{\tau_i} \MAT{A}\otimes \MAT{B}) \\ = \MAT{A}^{-1}\MAT{\xi}_{\MAT{A}} \otimes \MAT{I} +\MAT{I} \otimes \MAT{B}^{-1} \MAT{\xi}_{\MAT{B}} + \frac{\MAT{\xi}_{\tau_i}}{\tau_i} \MAT{I}.
	\end{multline*}
	By exploiting the relations $(\MAT{M}\otimes\MAT{N})(\MAT{O}\otimes\MAT{P})=\MAT{M}\MAT{O}\otimes\MAT{N}\MAT{P}$, $\tr(\MAT{M}\otimes\MAT{N})=\tr(\MAT{M})\tr(\MAT{N})$ and $\tr(\MAT{A}^{-1}\MAT{\xi}_{\MAT{A}})=\tr(\MAT{A}^{-1}\MAT{\eta}_{\MAT{A}})=0$ (and same property for matrix $\MAT{B}$), we obtain the announced result.
\end{proof}

\noindent
\textbf{Remark}: we notice that the determinant constraint is the best choice with respect to information geometry since it allows to obtain a separable metric. This is not the case when the trace constraint is used instead. 

One can notice that the metric of Proposition \ref{prop:fim} is separable into two scaled Riemannian affine invariant metrics on $\MAT{A}$ and $\MAT{B}$, respectively, and a last term on $\VEC{\tau}$. 
The geometry of a product manifold resulting from such a separable metric is simply obtained by combining the geometries corresponding to each component.
In particular, the Riemannian exponential mapping at $\theta\in\mathcal{M}$ is defined for $\xi\in T_{\theta}\mathcal{M}$ as
\begin{equation}
	\exp_\theta^{\mathcal{M}}(\xi) = \left( \MAT{A}\exp(\MAT{A}^{-1}\MAT{\xi}_{\MAT{A}}), \MAT{B}\exp(\MAT{B}^{-1}\MAT{\xi}_{\MAT{B}}) ,\VEC{\tau} \odot \exp(\VEC{\tau}^{\odot^{-1}}\odot \MAT{\xi}_{\VEC{\tau}}) \right).
	\label{eq:expm}
\end{equation}
This exponential mapping will be very useful when it comes to optimizing a cost function on the manifold. As for the exponential mapping, it is easy to derive the geodesic distance since the metric is completely separable, and we have:
\begin{equation}
		\delta_{\mathcal{M}}^2(\theta_0,\theta_1) = \frac{b}{p}\delta_{\mathcal{H}^{++}_a}^2(\MAT{A}_{0},\MAT{A}_{1}) +\frac{a}{p}\delta_{\mathcal{H}^{++}_b}^2(\MAT{B}_{0},\MAT{B}_{1}) + \frac{1}{n}\delta_{\mathbb{R}_{++}^n}^2(\VEC{\tau}_{0},\VEC{\tau}_{1}) ,
	\label{eq:defdistance}
\end{equation}
where $\delta_{\mathcal{H}^{++}_p}^2(\MAT{\Sigma}_{0},\MAT{\Sigma}_{1})=\| \logm(\MAT{\Sigma}_{0}^{-1/2}\MAT{\Sigma}_{1}\MAT{\Sigma}_{0}^{-1/2})\|^2_2$ and $\delta_{\mathbb{R}_{++}^n}^2(\VEC{\tau}_{0},\VEC{\tau}_{1})=\|\log(\VEC{\tau}_{0}^{-1}\odot\VEC{\tau}_{1})\|_2^2$.

Thanks to this geometry, we are now able to derive an online change detection when the covariance matrix possesses a Kronecker structure and the data follows a SG distribution.

\subsection{Online estimation of $\MAT{A}_{0}$, $\MAT{B}_{0}$ and $\VEC{\tau}_{0}$}

Given a new pixel patch $\mathcal{X}_{T+1} = \{\VEC{x}_i^{(T+1)}\}_{i \in \llbracket 1,n\rrbracket}$ in the time-series, the MLE of the parameter $\theta^{(0)}= \{ \MAT{A}^{(0)} , \MAT{B}^{(0)}, \VEC{\tau}^{(0)} \}$ is expressed as the fixed point \eqref{eq:MLE_H0}, that depends on the whole data $\mathcal{X}_{1:T+1} $.
In order to evaluate the proposed GLRT in a lightweight streaming fashion, our goal is to derive a recursive algorithm that provides an estimate of $\theta^{(0)}$ at time $T+1$, denoted $\hat{\theta}^{(0)}_{ \mathcal{X}_{T+1} }$, given solely the previous estimate $\hat{\theta}^{(0)}_{\mathcal{X}_{1:T} }$ and the incoming data $\mathcal{X}_{T+1}$.
Though many options can be envisioned to obtained such recursive algorithm, the interest of leveraging the stochastic natural gradient descent~\citep{zhou2019fast,bouchard2020riemannian} is twofold:
\begin{itemize}
    \item By leveraging the Riemannian gradient associated with the Fisher information metric (information geometry), a theoretically optimal gradient step can easily be obtained (hence, no line search will be involved in the update). Furthermore, this step only depends on intrinsic dimensions ($a$, $b$, $n$) of the estimation problem.
    \item The corresponding stochastic Riemannian gradient descent ensures statistical efficiency, i.e., that the recursive estimate $\hat{\theta}^{(0)}_{\mathcal{X}_{1:T+1}}$ tends to the MLE of $\theta^{(0)}$ under $\Hzero$ when $T \rightarrow \infty $, so the corresponding online GLRT is also asymptotically consistent.
\end{itemize}
An iteration of the proposed stochastic natural gradient descent~\citep{zhou2019fast,bouchard2020riemannian} is expressed as:
\begin{equation}
    \hat{\theta}^{(0)}_{ \mathcal{X}_{1:T+1} } 
    =
    \exp_{ 
    \hat{\theta}^{(0)}_{\mathcal{X}_{1:T}}
    }^{\mathcal{M}}
    \left(
        -\frac{\alpha_0}{T}\grad \mathcal{L}^{(T+1)}( \hat{\theta}^{(0)}_{\mathcal{X}_{1:T}} )
    \right)    
     \label{eq:KronSGD}
\end{equation}
where $\exp_{.}^{\mathcal{M}}$ is given in \eqref{eq:expm}, and where $\grad \mathcal{L}^{(T+1)}( \hat{\theta}^{(0)}_{\mathcal{X}_{1:T}})$
denotes the Riemannain gradient of the likelihood
of the data $\mathcal{X}_{T+1}$
evaluated at point $\hat{\theta}^{(0)}_{\mathcal{X}_{1:T}}$.
In order to apply this algorithm to the problem of interest, it remains to compute the Riemannian gradient of $\mathcal{L}^{(T+1)}$ according to the metric of Proposition~\ref{prop:fim}.
This Riemannian gradient is given in Proposition~\ref{prop:rgrad}.

\begin{Proposition}[Riemannian gradient]
	\label{prop:rgrad}
	The Riemannian gradient of $\mathcal{L}^{(T+1)}$ at $\theta\in\mathcal{M}$ according to the metric of Proposition~\ref{prop:fim} is given by
	\begin{multline*}
		\grad \mathcal{L}^{(T+1)}  (\theta) = 
		\left(
		-\sum_{i=1}^n\frac{1}{b\tau_i} P_{\MAT{A}}(\VEC{M}_i^T\MAT{B}^{-T}\MAT{M}_i^{*}) \right.
		\\\left.,-\sum_{i=1}^n\frac{1}{a\tau_i} P_{\MAT{B}}(\VEC{M}_i\MAT{A}^{-T}\MAT{M}_i^H) , n(\VEC{q}^{(T+1)}-p\VEC{\tau})		\right),
	\end{multline*}
	where 
 $q_i(\MAT{A},\MAT{B}) = \left.\VEC{x}_i^{(T+1)}\right.^H (\MAT{A} \otimes \MAT{B})^{-1} \VEC{x}^{(T+1)}_i$
 and
 $\VEC{q}^{(T+1)}=(q_1(\MAT{A},\MAT{B}),\ldots,q_n(\MAT{A},\MAT{B}))^T$.
 The matrix $P_{\MAT{A}}:\mathcal{H}_a\to T_{\MAT{A}}s\mathcal{H}^{++}_a$ is the orthogonal projection map such that 
	\begin{equation*}
		P_{\MAT{A}}(\MAT{\xi}_{\MAT{A}}) = \MAT{\xi}_{\MAT{A}} - \frac{\tr(\MAT{A}^{-1}\MAT{\xi}_{\MAT{A}})}a\MAT{A}.
	\end{equation*}
 and the matrix $P_{\MAT{B}}(\MAT{\xi}_{\MAT{B}})$ follows from the fame expression.
\end{Proposition}

\begin{proof}
The likelihood of the new dataset $\{\VEC{x}_i^{(T+1)}\}_{i \in \llbracket 1,n\rrbracket}$ is
\begin{equation}
	\mathcal{L}^{(T+1)}(\theta^{(T+1)}) = \sum_{i=1}^n \ell_i(\theta^{(T+1)}),
\end{equation}
where
\begin{equation*}
	\ell_i(\theta) = -p\log(\tau_i) - \frac{q_i(\MAT{A},\MAT{B})}{\tau_i},
\end{equation*}
with $q_i(\MAT{A},\MAT{B}) = \left.\VEC{x}_i^{(T+1)}\right.^H (\MAT{A} \otimes \MAT{B})^{-1} \VEC{x}^{(T+1)}_i$.
Notice that the log-determinant terms in the likelihood vanish because we consider matrices of determinant one.
    The Riemannian gradient is defined according to the metric through the relation:
    \begin{equation}
        \langle \grad \mathcal{L}^{(T+1)}(\theta),\MAT{\xi} \rangle_{\theta}^{\mathcal{M}} = \D \mathcal{L}^{(T+1)}(\theta)[\MAT{\xi}].
        \label{eq:defgrad}
    \end{equation}
    Thus we need the directional derivative of $\mathcal{L}^{(T+1)}(\theta)$:
    \begin{equation}
        \begin{array}{lll}
             \D \mathcal{L}^{(T+1)}(\theta)[\MAT{\xi}] & = &  -\sum_{i=1}^n\left(\frac{p\MAT{\xi}_{\tau_i}}{\tau_i}+\frac{1}{\tau_i^2}q_i(\MAT{A},\MAT{B})\MAT{\xi}_{\tau_i}\right) 
               - \sum_{i=1}^n\frac{\D q_i(\MAT{A},\MAT{B})[\xi]}{\tau_i} \vspace{0.2cm}\\
             & = &  \frac{1}{n} \langle n(\VEC{q}^{(t)}-p\VEC{\tau}),\MAT{\xi}_{\VEC{\tau}} \rangle_{\tau}^{\mathbb{R}^n_{++}}  - \sum_{i=1}^n\frac{\D q_i(\MAT{A},\MAT{B})[\xi]}{\tau_i}
        \end{array}
    \end{equation}
    where $\langle \MAT{\xi}_{\VEC{\tau}},\MAT{\eta}_{\VEC{\tau}} \rangle_{\tau}^{\mathbb{R}^n_{++}} = (\MAT{\xi}_{\VEC{\tau}}\odot \VEC{\tau}^{\odot^{-1}})^T(\MAT{\eta}_{\VEC{\tau}}\odot \VEC{\tau}^{\odot^{-1}})$. The second term can be obtained thanks to the following result:
    \begin{equation}
        \begin{array}{cc}
             g(\MAT{X}) = \tr(\MAT{X} \VEC{x}\VEC{x}^H), & \D g(\MAT{X})[\MAT{\xi}] = \tr(\MAT{\xi}\VEC{x}\VEC{x}^H) \\
             g(\MAT{X}) = \tr(\MAT{X}^{-1} \VEC{x}\VEC{x}^H), & \D g(\MAT{X})[\MAT{\xi}] = \tr(\MAT{X}^{-1}\MAT{\xi}\MAT{X}^{-1}\VEC{x}\VEC{x}^H) 
        \end{array}
    \end{equation}
    which leads to
    \begin{equation}
        \begin{array}{lll}
             \D q_i(\MAT{A},\MAT{B})[\xi] & = & - \tr(\MAT{A}^{-1}\MAT{\xi}_{\MAT{A}}\MAT{A}^{-1}\MAT{M}_i^T\MAT{B}^{-T}\MAT{M}_i^{*}) \\ & &-\tr(\MAT{B}^{-1}
             \MAT{\xi}_{\MAT{B}}\MAT{B}^{-1}\MAT{M}_i\MAT{A}^{-T}\MAT{M}_i^{*}) \\
             & = & -\langle \MAT{M}_i^T\MAT{B}^{-T}\MAT{M}_i^{*},\MAT{\xi}_{\MAT{A}}\rangle_{\MAT{A}}^{\mathcal{H}_a^{++}} -\langle \MAT{M}_i\MAT{A}^{-T}\MAT{M}_i^H,\MAT{\xi}_{\MAT{B}}\rangle_{\MAT{B}}^{\mathcal{H}_b^{++}}
        \end{array}
    \end{equation}
    Finally, the result of the Riemannian gradient according to the metric of Proposition~\ref{prop:fim} is obtained through identification by \eqref{eq:defgrad} and projection onto the tangent spaces.
\end{proof}

Thanks to this proposition, the online GLRT is finally computed by plugging the estimates updates from \eqref{eq:KronSGD} in the detector \eqref{eq:GLRTKronCG}. 

To evaluate the quality of the estimation of parameters $\MAT{A}_{0}$, $\MAT{B}_{0}$ and $\VEC{\tau}_{0}$, it is interesting to derive the intrinsic Cramér-Rao bound (ICRB) \citep{Smith2005,bouchard2020riemannian} and the corresponding inequality. This is the purpose of the next proposition.
\begin{Proposition}[ICRB]
	Given an unbiased estimator $\hat{\theta}^{(T)}$ of $\hat{\theta}_{\mathcal{X}_{1:T}}^{(T)}$ corresponding to a MITS with $T$ data, the ICRB corresponding to the error measured with the Riemannian distance \eqref{eq:defdistance} and the metric given in the proposition \eqref{prop:fim} is
	\begin{equation*}
		\mathbb{E}[\delta^2_{\mathcal{M}}(\theta^{(T)},\hat{\theta}_{\mathcal{X}_{1:T}}^{(T)})] \leq \frac{(a^2-1)+(b^2-1)+n}{Tpn}
	\end{equation*}
\end{Proposition}
\begin{proof}
Since the distance is obtained trough the Fisher information metric, the ICRB for the problem of estimation of $\theta^{(T)}$ could be easily deduced from the approach of \citep{Smith2005,bouchard2020riemannian} and is the ratio between the dimension of the problems and the number of parameters to estimate. 
\end{proof}

Thanks to the fact that the metric is separable, it is then easy to obtain the following inequalities for each component of $\theta^{(T)}$:
\begin{equation}
    \begin{array}{lll}
      \mathbb{E}[\delta^2_{\mathcal{H}^{++}_a}(\hat{\MAT{A}}_{\mathcal{X}_{1:T}}^{(T)},\MAT{A}^{(T)})]   & \leq & \frac{(a^2-1)}{bTn},  \\
      \mathbb{E}[\delta^2_{\mathcal{H}^{++}_b}(\hat{\MAT{B}}_{\mathcal{X}_{1:T}}^{(T)},\MAT{B}^{(T)})]    & \leq & \frac{(b^2-1)}{aTn}, \\
      \mathbb{E}[\delta^2_{\mathbb{R}_{++}^n}(\hat{\VEC{\tau}}_{\mathcal{X}_{1:T}}^{(T)},\VEC{\tau}^{(T)})]    & \leq & \frac{1}{Tp}.
    \end{array}
    \label{eq:ICRB_param}
\end{equation}

%% file: sections/4_numerical_results.tex
In this section, we first show the performances in terms of MSE of the online estimator proposed in \ref{sec:geo}. Next, we will show the interest of this online approach to the problem of change detection with simulated data as well as with an experiment composed of real data.

\subsection{Estimation performance}

In this section, we compare the performance of the maximum likelihood estimator obtained with a usual Riemannian optimization algorithm and of the one obtained with the on-line procedure of \eqref{eq:KronSGD}.
In order to do so, simulated data drawn from the multivariate K-distribution with $\nu=1$ for the shape parameter are generated with a Kronecker product structured covariance. To build a $p\times p$ ($p=12$) covariance matrix, we compute
\begin{equation}
	\begin{array}{c}
		\MAT{\Sigma} = \MAT{A} \otimes \MAT{B},
		\\
		\MAT{A} = \MAT{U}_{\MAT{A}} \MAT{\Lambda}_{\MAT{A}} \MAT{U}_{\MAT{A}}^T,
		\qquad
		\MAT{B} = \MAT{U}_{\MAT{B}} \MAT{\Lambda}_{\MAT{B}} \MAT{U}_{\MAT{B}}^T,
	\end{array}
\end{equation}
where $a=4$ and $b=3$,
\begin{itemize}
	\item $\MAT{U}_{\MAT{A}}$ and $\MAT{U}_{\MAT{B}}$ are random orthogonal matrices,
	\item $\MAT{\Lambda}_{\MAT{A}}$ and $\MAT{\Lambda}_{\MAT{B}}$ are diagonal matrices whose minimal and maximal elements are $\nicefrac{1}{\sqrt{c}}$ and $\sqrt{c}$ ($c=10$ is the condition number with respect to inversion); their other elements are randomly drawn from the uniform distribution between $\nicefrac{1}{\sqrt{c}}$ and $\sqrt{c}$; the determinant of $\MAT{\Lambda}_{\MAT{A}}$ is then normalized.
\end{itemize}
The number of samples is fixed to $n=8$ since the conclusions keep the same for larger $n$. The number of frames $T$ varies from 1 to 1000. 1000 trials are used to estimate the MSEs. 

For this experiment, we consider the following estimators:
\begin{itemize}
	\item the classical maximum-likelihood estimator obtained with Riemannian gradient descent (GD). Optimization for this estimator is performed with pymanopt toolbox \citep{manopt,pymanopt}.
	\item the online version obtained through stochastic gradient descent (SGD) of \eqref{eq:KronSGD}.
\end{itemize}
The online algorithm is initialized with $\theta_0=(\MAT{A}^{(1)},\MAT{B}^{(1)},\mathbf{\tau}^{(1)})$ obtained with the GD when $T=1$. 

In order to measure the performance of the estimators, the distances defined in \eqref{eq:defdistance}, $\delta^2_{\mathcal{H}^{++}_a}$, $\delta^2_{\mathcal{H}^{++}_a}$ and $\delta^2_{\mathbb{R}_{++}^n}$, will be used. Moreover, the ICRBs for each term given in \eqref{eq:ICRB_param} are also computed.

The results are presented in Figure~\ref{fig:mse_AB} for the MSEs of $\MAT{A}$ and $\MAT{B}$ and in \ref{fig:mse_tau} for the MSE of $\VEC{\tau}$. First, we notice that the on-line version converges to the classical estimate which is expected as stated in \citep{zhou2019fast}. We notice that for a value of $T=100$ the results of both algorithms become very close which is clearly interesting in terms of practical interest in applications like change detection.

\begin{figure}[t]
\begin{center}
\begin{subfigure}[b]{0.65\textwidth}
	\input{./figures/MSE_A.tex}
\end{subfigure}
\begin{subfigure}[b]{0.65\textwidth}
	\input{./figures/MSE_B.tex}
\end{subfigure}
	\caption{
		Mean of error measures on $\MAT{A}$ (left) and $\MAT{B}$ (left) of the classical gradient descent method (GD) and its on-line counterpart (SGD) as functions of the number of frames $T$. Means are computed over $1000$ simulated sets $\{\VEC{x}_i\}_{i=1}^n$ with $\nu=1$. The data size is $p=12$ with $a=4$ and $b=3$ and the number of samples is $n=8$.
	}
\label{fig:mse_AB}
 \end{center}
\end{figure}

\begin{figure}[t]
\centering
	\input{./figures/MSE_tau.tex}
	\caption{
		Mean of error measures on $\boldsymbol{\tau}$ of the classical gradient descent method (GD) and its on-line counterpart (SGD) as functions of the number of frames $T$. Means are computed over $1000$ simulated sets $\{\VEC{x}_i\}_{i=1}^n$ with $\nu=1$. The data size is $p=12$ with $a=4$ and $b=3$ and the number of samples is $n=8$.
	}
\label{fig:mse_tau}
\end{figure}

\subsection{Performance in change detection with simulated data}

In this section, we compare the ROC plots of the change detection detector of \eqref{eq:GLRTKronCG} in the offline and online versions. We hope that for a large value of $T$, the performance become close. Like in the previous subsection, $p=12$ with $a=4$ and $b=3$. We also simulate a K-distribution for the data. But, the covariance matrices are differently built to be better fitted with the change detection application. In both hypothesis, $\MAT{A}$ and $\MAT{B}$ are Toeplitz matrices of correlation coefficient $\rho_0$ and $\rho_1$ (for $\mathbf{A}$: $\rho_0 = 0.3+0.7j$, $\rho_1=0.3+0.5j$. For $\mathbf{B}$: $\rho_0 = 0.3+0.6j$, $\rho_1=0.4+0.5j$). For the $H_1$ hypothesis, the change is put at the middle of the time series, i.e. at $T/2$. The number of samples $n$ is 13. To show the interest to take into account of the Kronecker structure in the change detection algorithm, we compare our two versions of the Kronecker change detection to the one proposed in \citep{Mian2019TSP} and given in \eqref{eq:GLRTMatGen}. To resume, we test 4 change detectors, $\hat{\Lambda}_{SG}^{(T)}$ and $\hat{\Lambda}_{K-SG}^{(T)}$, and their corresponding on-line versions, $\hat{\Lambda}_{SG-O}^{(T)}$ and $\hat{\Lambda}_{K-SG-O}^{(T)}$. For the offline versions, the number of images $T$ is equal to 50 whereas the online versions are built from $T$ varying from 2 to 50. All ROC plots are estimated by using 5000 trials. 

Figure \ref{fig:ROC_Gaussian} shows the ROC plots when the data follows a Gaussian distribution. Figures \ref{fig:ROC_nu100} and \ref{fig:ROC_nu1} show the same plots but when the data follow a K distribution with $\nu=100$ and $\nu=1$ as shape parameters respectively. We notice that all the Kronecker versions outperform the classical ones which is expected by the structure of the covariance matrix.  For the heterogeneous hypothesis ($\nu=1$), we get a very good result, since online versions perform similarly to offline detection when $T=50$. Conversely, we have the surprising result that online versions have a slower convergence speed for Gaussian or quasi-Gaussian data. Nevertheless, we can conclude that the detection results of the online detectors improve with $T$ and converge towards the offline result, particularly when the Kronecker structure is used. 

\begin{figure}[h]
    \centering
    \begin{subfigure}[b]{0.45\textwidth}
        \input{figures/roc-simulated/gaussian/Scaled_Gaussian_GLRT_vs_Scaled_Gaussian_SGD}
    \end{subfigure}
    \hfill
    \begin{subfigure}[b]{0.5\textwidth}
        \input{figures/roc-simulated/gaussian/Scaled_Gaussian_Kronecker_GLRT_vs_Scaled_Gaussian_Kronecker_SGD}
    \end{subfigure}
    \caption{Gaussian data with $a=3$, $b=4$, $n=13$. For $\mathbf{A}$: $\rho_0 = 0.3+0.7j$, $\rho_1=0.3+0.5j$. For $\mathbf{B}$: $\rho_0 = 0.3+0.6j$, $\rho_1=0.4+0.5j$. 5000 trials.}
    \label{fig:ROC_Gaussian}
\end{figure}
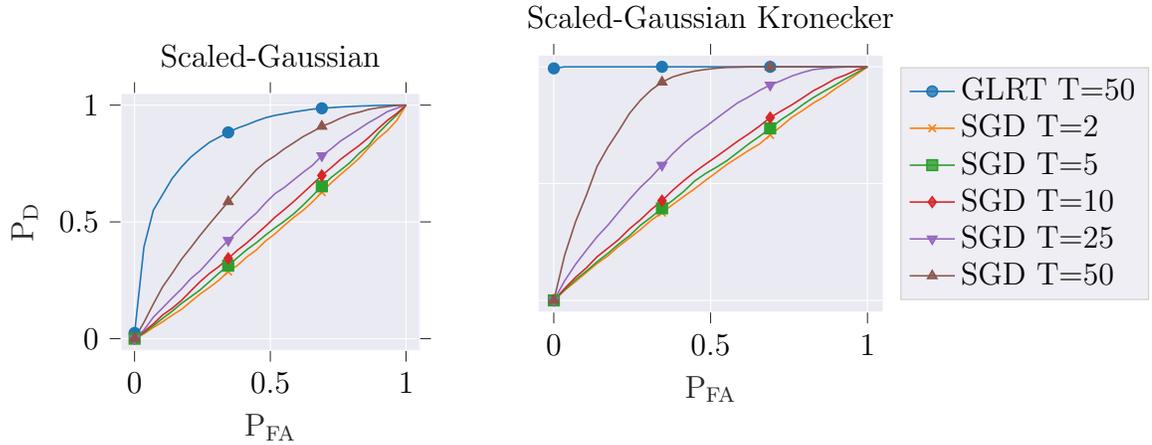

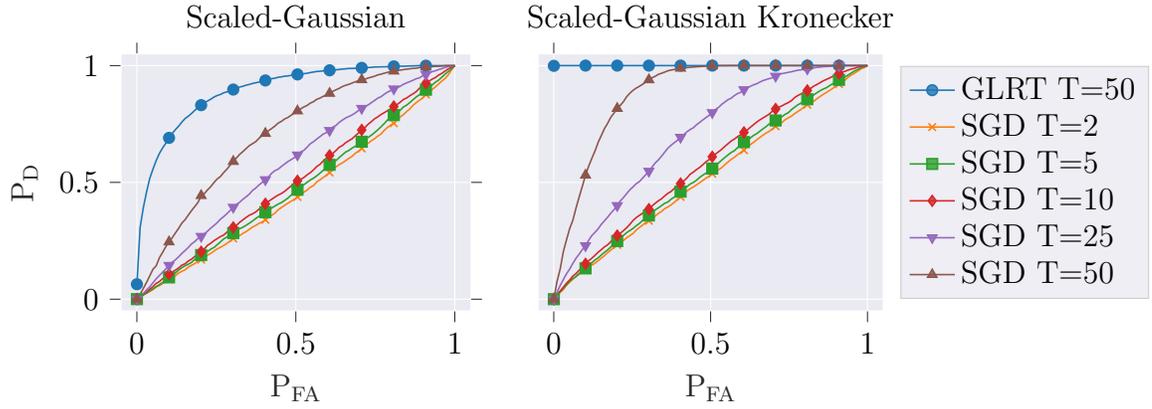
\begin{figure}[h]
    \centering
    \begin{subfigure}[b]{0.45\textwidth}
        \input{figures/roc-simulated/pseudo-gaussian/Scaled_Gaussian_GLRT_vs_Scaled_Gaussian_SGD}
    \end{subfigure}
    \hfill
    \begin{subfigure}[b]{0.5\textwidth}
        \input{figures/roc-simulated/pseudo-gaussian/Scaled_Gaussian_Kronecker_GLRT_vs_Scaled_Gaussian_Kronecker_SGD}
    \end{subfigure}
    \caption{Non-Gaussian data with $a=3$, $b=4$, $n=13$, $\nu=100$. For $\mathbf{A}$: $\rho_0 = 0.3+0.7j$, $\rho_1=0.3+0.5j$. For $\mathbf{B}$: $\rho_0 = 0.3+0.6j$, $\rho_1=0.4+0.5j$. 5000 trials.}
    \label{fig:ROC_nu100}
\end{figure}

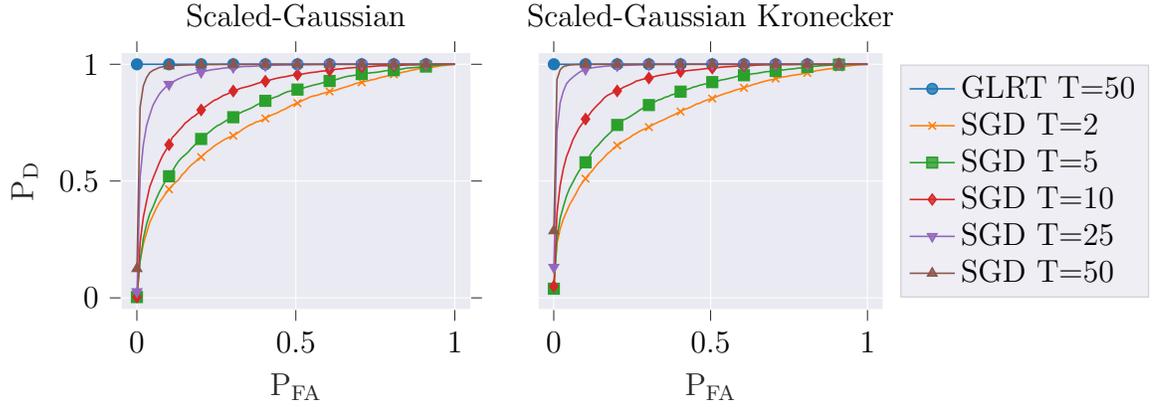
\begin{figure}[h]
    \centering
    \begin{subfigure}[b]{0.45\textwidth}
        \input{figures/roc-simulated/non-gaussian/Scaled_Gaussian_GLRT_vs_Scaled_Gaussian_SGD}
    \end{subfigure}
    \hfill
    \begin{subfigure}[b]{0.5\textwidth}
        \input{figures/roc-simulated/non-gaussian/Scaled_Gaussian_Kronecker_GLRT_vs_Scaled_Gaussian_Kronecker_SGD}
    \end{subfigure}
    \caption{Non-Gaussian data with $a=3$, $b=4$, $n=13$, $\nu=1$. For $\mathbf{A}$: $\rho_0 = 0.3+0.7j$, $\rho_1=0.3+0.5j$. For $\mathbf{B}$: $\rho_0 = 0.3+0.6j$, $\rho_1=0.4+0.5j$. 5000 trials.}
    \label{fig:ROC_nu1}
\end{figure}

\subsection{Performance in change detection with real data}

The SAR ITS data set is taken from UAVSAR (courtesy of NASA/JPL-Caltech) and is referenced as SDelta\_28518\_02, Segment 1. The number of images available is $T=68$. They show the evolution of a river delta in the USA with cycle of droughts and flood. Some strong scatterers appear and disappear over time. Since the image cover an extensive area of the delta, we crop it to an interesting part of size $200\times 200$ pixels. The SAR images correspond to full-polarization data with a resolution of 1.67 m in range and 0.6 m in azimuth. Thanks to this high resolution property of the SAR images, the scatterers present in this scene exhibit an interesting spectro-angular behavior, each polarization of these images has been subjected to the wavelet transform presented in [5], allowing to obtain images of dimension $p = 12$. In this particular configuration, the full covariance matrices show an inherent Kronecker structure $\mathbf{A} \otimes \mathbf{B}$. The matrix $\mathbf{A} \in \mathbb{C}^{a \times a}$ corresponds to the spectro-angular property ($a=4$) although $\mathbf{B} \in \mathbb{C}^{b \times b}$ is linked to the polarization ($b=3$).

As in the previous section, we test 4 change detectors, $\hat{\Lambda}_{SG}^{(T)}$ and $\hat{\Lambda}_{K-SG}^{(T)}$, and their corresponding on-line versions, $\hat{\Lambda}_{SG-O}^{(T)}$ and $\hat{\Lambda}_{K-SG-O}^{(T)}$. Figure \ref{fig:RealData_NoRepet} shows the outputs of the four detectors for the complete time series. Since we do not have any ground truths we are not able to conclude if the Kronecker versions are better than the classical ones. But we are interested in the similarity between the online and offline versions. For $T=68$, it is difficult to conclude since both results are quite different. We propose then to repeat the time series in order to have a larger one. Figure \ref{fig:RealData_Repet5} shows the result when the number of repetition is 5 ($T=340$). For this time series, we find better similarity between the online and corresponding offline versions, especially when the Kronecker structure is used in the change detector. This behavior is better illustrated in Figure \ref{fig:RealData_Repet10} when the number of repetitions is 10 ($T=680$).

\begin{figure*}[t]
    \centering
    \begin{subfigure}[b]{0.4\textwidth}
        \input{figures/realdata/repeat-none/scaled-gaussian-glrt}
        \caption{Scaled-Gaussian GLRT}
    \end{subfigure}
        \begin{subfigure}[b]{0.4\textwidth}
        \input{figures/realdata/repeat-none/scaled-gaussian-sgd}
        \caption{Scaled-Gaussian SGD}
    \end{subfigure}
    \begin{subfigure}[b]{0.4\textwidth}
        \input{figures/realdata/repeat-none/scaled-gaussian-kron-glrt}
        \caption{Scaled-Gaussian Kronecker GLRT}
    \end{subfigure}
    \begin{subfigure}[b]{0.4\textwidth}
        \input{figures/realdata/repeat-none/scaled-gaussian-kron-sgd}
        \caption{Scaled-Gaussian Kronecker SGD}
    \end{subfigure}
    \caption{Window size of $7\times 7$. No repetitions of time series: $T=68$.}
    \label{fig:RealData_NoRepet}
\end{figure*}
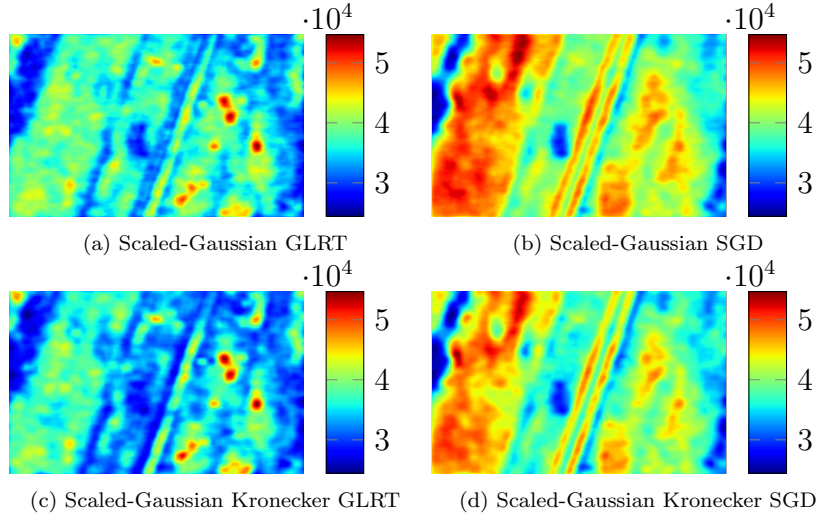

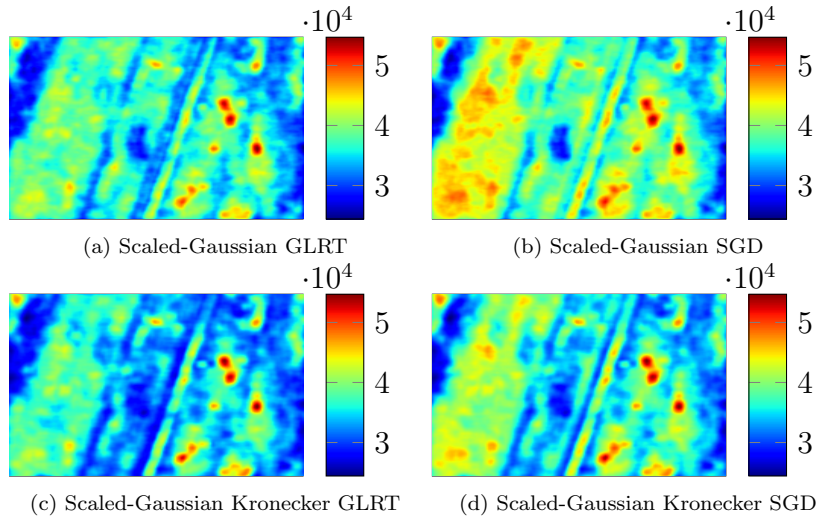
\begin{figure*}[t]
    \centering
    \begin{subfigure}[b]{0.4\textwidth}
        \input{figures/realdata/repeat-5/scaled-gaussian-glrt}
        \caption{Scaled-Gaussian GLRT}
    \end{subfigure}
        \begin{subfigure}[b]{0.4\textwidth}
        \input{figures/realdata/repeat-5/scaled-gaussian-sgd}
        \caption{Scaled-Gaussian SGD}
    \end{subfigure}
    \begin{subfigure}[b]{0.4\textwidth}
        \input{figures/realdata/repeat-5/scaled-gaussian-kron-glrt}
        \caption{Scaled-Gaussian Kronecker GLRT}
    \end{subfigure}
    \begin{subfigure}[b]{0.4\textwidth}
        \input{figures/realdata/repeat-5/scaled-gaussian-kron-sgd}
        \caption{Scaled-Gaussian Kronecker SGD}
    \end{subfigure}
    \caption{Window size of $7\times 7$. $5$ repetitions of time series: $T=340$.}
    \label{fig:RealData_Repet5}
\end{figure*}

\begin{figure*}[t]
    \centering
    \begin{subfigure}[b]{0.4\textwidth}
        \input{figures/realdata/repeat-10/scaled-gaussian-glrt}
        \caption{Scaled-Gaussian GLRT}
    \end{subfigure}
        \begin{subfigure}[b]{0.4\textwidth}
        \input{figures/realdata/repeat-10/scaled-gaussian-sgd}
        \caption{Scaled-Gaussian SGD}
    \end{subfigure}
    \begin{subfigure}[b]{0.4\textwidth}
        \input{figures/realdata/repeat-10/scaled-gaussian-kron-glrt}
        \caption{Scaled-Gaussian Kronecker GLRT}
    \end{subfigure}
    \begin{subfigure}[b]{0.4\textwidth}
        \input{figures/realdata/repeat-10/scaled-gaussian-kron-sgd}
        \caption{Scaled-Gaussian Kronecker SGD}
    \end{subfigure}
    \caption{Window size of $7\times 7$. $10$ repetitions of time series: $T=680$.}
    \label{fig:RealData_Repet10}
\end{figure*}
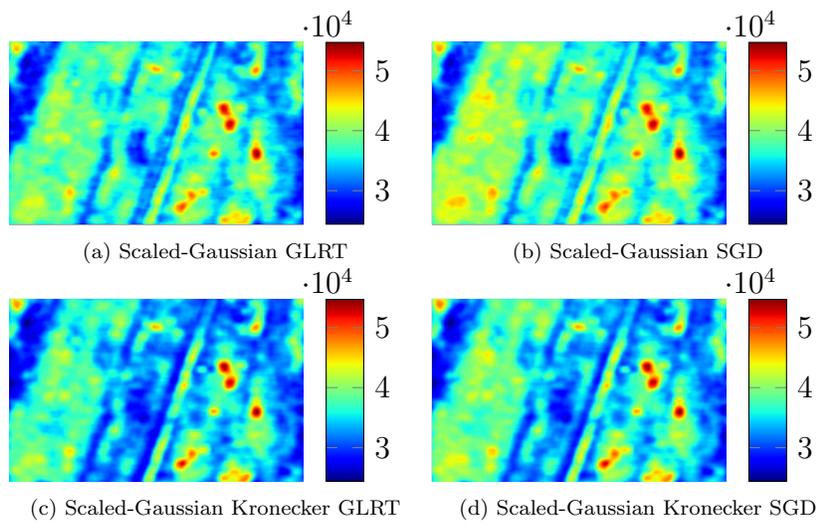

%% file: figures/MSE_A.tex
\begin{tikzpicture}

\definecolor{darkorange25512714}{RGB}{255,127,14}
\definecolor{darkslategray38}{RGB}{38,38,38}
\definecolor{forestgreen4416044}{RGB}{44,160,44}
\definecolor{lavender234234242}{RGB}{234,234,242}
\definecolor{lightgray204}{RGB}{204,204,204}
\definecolor{steelblue31119180}{RGB}{31,119,180}

\begin{axis}[
axis background/.style={fill=lavender234234242},
axis line style={white},
legend cell align={left},
legend style={
  fill opacity=0.8,
  draw opacity=1,
  text opacity=1,
  draw=lightgray204,
  fill=lavender234234242
},
log basis x={10},
log basis y={10},
tick align=outside,
title={MSE on estimation of \(\displaystyle \mathbf{A}\). \(\displaystyle a=4\), \(\displaystyle b=3\).},
x grid style={white},
xlabel=\textcolor{darkslategray38}{\(\displaystyle T\)},
xmajorgrids,
xmin=0.732911438554269, xmax=682.210665160709,
xmode=log,
xtick style={color=darkslategray38},
y grid style={white},
ylabel=\textcolor{darkslategray38}{MSE},
ymajorgrids,
ymin=0.000926209215222125, ymax=0.642261817427857,
ymode=log,
ytick style={color=darkslategray38}
]
\addplot [semithick, steelblue31119180]
table {%
1 0.3570162127022
3 0.167709809677326
5 0.108860171941678
10 0.0587673185559684
20 0.0302632433831038
50 0.0124171541084058
100 0.00615902981371699
400 0.00156090621502931
500 0.00124694011841292
};
\addlegendentry{Gradient Kronecker $\mathbf{A}$}
\addplot [semithick, darkorange25512714]
table {%
1 0.477062855788239
3 0.213643848250412
5 0.133080902263682
10 0.0676113098807557
20 0.0332032053998289
50 0.013026314783867
100 0.00632395886706776
400 0.00157563350108662
500 0.00126050093592433
};
\addlegendentry{Stochastic Kronecker $\mathbf{A}$}
\addplot [semithick, forestgreen4416044]
table {%
1 0.3125
3 0.15625
5 0.104166666666667
10 0.0568181818181818
20 0.0297619047619048
50 0.0122549019607843
100 0.00618811881188119
400 0.00155860349127182
500 0.00124750499001996
};
\addlegendentry{ICRB $\mathbf{A}$}
\end{axis}

\end{tikzpicture}

%% file: figures/MSE_B.tex
\begin{tikzpicture}

\definecolor{darkorange25512714}{RGB}{255,127,14}
\definecolor{darkslategray38}{RGB}{38,38,38}
\definecolor{forestgreen4416044}{RGB}{44,160,44}
\definecolor{lavender234234242}{RGB}{234,234,242}
\definecolor{lightgray204}{RGB}{204,204,204}
\definecolor{steelblue31119180}{RGB}{31,119,180}

\begin{axis}[
axis background/.style={fill=lavender234234242},
axis line style={white},
legend cell align={left},
legend style={
  fill opacity=0.8,
  draw opacity=1,
  text opacity=1,
  draw=lightgray204,
  fill=lavender234234242
},
log basis x={10},
log basis y={10},
tick align=outside,
title={MSE on estimation of \(\displaystyle \mathbf{B}\). \(\displaystyle a=4\), \(\displaystyle b=3\).},
x grid style={white},
xlabel=\textcolor{darkslategray38}{T},
xmajorgrids,
xmin=0.732911438554269, xmax=682.210665160709,
xmode=log,
xtick style={color=darkslategray38},
y grid style={white},
ylabel=\textcolor{darkslategray38}{MSE},
ymajorgrids,
ymin=0.000369049596409876, ymax=0.281495069586041,
ymode=log,
ytick style={color=darkslategray38}
]
\addplot [semithick, steelblue31119180]
table {%
1 0.143399281474524
3 0.066573214999724
5 0.0432987183753161
10 0.0230782374022517
20 0.0122039867987429
50 0.00489779567637116
100 0.00251799594453833
400 0.000636134709094535
500 0.000508833430212024
};
\addlegendentry{Gradient Kronecker $\mathbf{B}$}
\addplot [semithick, darkorange25512714]
table {%
1 0.208186826211485
3 0.0904897313670738
5 0.0557896662841426
10 0.0278103560812835
20 0.0137224369266994
50 0.00519735020444157
100 0.00262039766929099
400 0.000643171011222107
500 0.00051455546977885
};
\addlegendentry{Stochastic Kronecker $\mathbf{B}$}
\addplot [semithick, forestgreen4416044]
table {%
1 0.125
3 0.0625
5 0.0416666666666667
10 0.0227272727272727
20 0.0119047619047619
50 0.00490196078431373
100 0.00247524752475248
400 0.000623441396508728
500 0.000499001996007984
};
\addlegendentry{ICRB $\mathbf{B}$}
\end{axis}

\end{tikzpicture}

%% file: figures/MSE_tau.tex
\begin{tikzpicture}

\definecolor{darkorange25512714}{RGB}{255,127,14}
\definecolor{darkslategray38}{RGB}{38,38,38}
\definecolor{forestgreen4416044}{RGB}{44,160,44}
\definecolor{lavender234234242}{RGB}{234,234,242}
\definecolor{lightgray204}{RGB}{204,204,204}
\definecolor{steelblue31119180}{RGB}{31,119,180}

\begin{axis}[
axis background/.style={fill=lavender234234242},
axis line style={white},
legend cell align={left},
legend style={
  fill opacity=0.8,
  draw opacity=1,
  text opacity=1,
  draw=lightgray204,
  fill=lavender234234242
},
log basis x={10},
log basis y={10},
tick align=outside,
title={MSE on estimation of \(\displaystyle \boldsymbol{\tau}\). \(\displaystyle a=4\), \(\displaystyle b=3\).},
x grid style={white},
xlabel=\textcolor{darkslategray38}{\(\displaystyle T\)},
xmajorgrids,
xmin=0.732911438554269, xmax=682.210665160709,
xmode=log,
xtick style={color=darkslategray38},
y grid style={white},
ylabel=\textcolor{darkslategray38}{MSE},
ymajorgrids,
ymin=0.000122711229199615, ymax=0.0988580824438407,
ymode=log,
ytick style={color=darkslategray38}
]
\addplot [semithick, steelblue31119180]
table {%
1 0.0560505707915002
3 0.0245356507029833
5 0.015301247547722
10 0.00804763136970804
20 0.00414607266202801
50 0.00167381447236113
100 0.000851406616911022
400 0.000213971175328741
500 0.000171799889954598
};
\addlegendentry{Gradient Kronecker $\boldsymbol{\tau}$}
\addplot [semithick, darkorange25512714]
table {%
1 0.0729315528397592
3 0.0373573102780496
5 0.0239226135767275
10 0.0119899649794041
20 0.00572256553432043
50 0.00204271878113437
100 0.000970682111700002
400 0.000224948699550232
500 0.000179686592911959
};
\addlegendentry{Stochastic Kronecker $\boldsymbol{\tau}$}
\addplot [semithick, forestgreen4416044]
table {%
1 0.0416666666666667
3 0.0208333333333333
5 0.0138888888888889
10 0.00757575757575758
20 0.00396825396825397
50 0.00163398692810458
100 0.000825082508250825
400 0.000207813798836243
500 0.000166333998669328
};
\addlegendentry{ICRB $\boldsymbol{\tau}$}
\end{axis}

\end{tikzpicture}

%% file: figures/roc-simulated/gaussian/Scaled_Gaussian_GLRT_vs_Scaled_Gaussian_SGD.tex
\begin{tikzpicture}

\definecolor{crimson2143940}{RGB}{214,39,40}
\definecolor{darkorange25512714}{RGB}{255,127,14}
\definecolor{darkslategray38}{RGB}{38,38,38}
\definecolor{forestgreen4416044}{RGB}{44,160,44}
\definecolor{lavender234234242}{RGB}{234,234,242}
\definecolor{lightgray204}{RGB}{204,204,204}
\definecolor{mediumpurple148103189}{RGB}{148,103,189}
\definecolor{sienna1408675}{RGB}{140,86,75}
\definecolor{steelblue31119180}{RGB}{31,119,180}

\begin{axis}[
axis background/.style={fill=lavender234234242},
axis line style={white},
height=5cm,
every axis legend/.code={\let\addlegendentry\relax},
tick align=outside,
width=0.9\columnwidth,
x grid style={white},
xlabel=\textcolor{darkslategray38}{\(\displaystyle \mathrm{P}_{\mathrm{FA}}\)},
xmajorgrids,
xmajorticks=true,
xmin=-0.04979, xmax=1.04999,
xtick style={color=darkslategray38},
y grid style={white},
ylabel=\textcolor{darkslategray38}{\(\displaystyle \mathrm{P}_{\mathrm{D}}\)},
ymajorgrids,
ymajorticks=true,
ymin=-0.05, ymax=1.05,
ytick style={color=darkslategray38},
title={Scaled-Gaussian}
]
\addplot [semithick, steelblue31119180, mark=*, mark size=2, mark repeat=10, mark options={solid}]
table {%
0.0002 0.0246
0.0346 0.3928
0.0692 0.55
0.1036 0.6176
0.1382 0.6852
0.1726 0.735
0.207 0.7774
0.2416 0.81
0.276 0.8408
0.3104 0.8622
0.345 0.8832
0.3794 0.9018
0.414 0.9162
0.4484 0.9322
0.4828 0.946
0.5174 0.9562
0.5518 0.963
0.5862 0.9706
0.6208 0.9766
0.6552 0.9826
0.6898 0.9864
0.7242 0.989
0.7586 0.992
0.7932 0.9938
0.8276 0.996
0.862 0.9972
0.8966 0.9986
0.931 0.9992
0.9656 1
1 1
};
\addlegendentry{GLRT T=50}
\addplot [semithick, darkorange25512714, mark=x, mark size=2, mark repeat=10, mark options={solid}]
table {%
0.0002 0.0002
0.0346 0.0206
0.0692 0.048
0.1036 0.0724
0.1382 0.1004
0.1726 0.1252
0.207 0.1586
0.2416 0.1914
0.276 0.2196
0.3104 0.2498
0.345 0.2872
0.3794 0.314
0.414 0.3504
0.4484 0.3796
0.4828 0.4186
0.5174 0.449
0.5518 0.4854
0.5862 0.519
0.6208 0.5496
0.6552 0.5858
0.6898 0.627
0.7242 0.6672
0.7586 0.7028
0.7932 0.737
0.8276 0.7728
0.862 0.8168
0.8966 0.8574
0.931 0.8906
0.9656 0.9342
1 1
};
\addlegendentry{SGD T=2}
\addplot [semithick, forestgreen4416044, mark=square*, mark size=2, mark repeat=10, mark options={solid}]
table {%
0.0002 0
0.0346 0.025
0.0692 0.0556
0.1036 0.0872
0.1382 0.1174
0.1726 0.1508
0.207 0.1796
0.2416 0.209
0.276 0.2442
0.3104 0.2804
0.345 0.3126
0.3794 0.3476
0.414 0.3816
0.4484 0.4106
0.4828 0.4438
0.5174 0.4746
0.5518 0.5032
0.5862 0.5346
0.6208 0.5718
0.6552 0.6116
0.6898 0.6524
0.7242 0.6904
0.7586 0.7214
0.7932 0.7538
0.8276 0.794
0.862 0.8276
0.8966 0.8786
0.931 0.9194
0.9656 0.9574
1 1
};
\addlegendentry{SGD T=5}
\addplot [semithick, crimson2143940, mark=diamond*, mark size=2, mark repeat=10, mark options={solid}]
table {%
0.0002 0.0004
0.0346 0.03
0.0692 0.0634
0.1036 0.1018
0.1382 0.13
0.1726 0.1634
0.207 0.2042
0.2416 0.2442
0.276 0.2794
0.3104 0.3086
0.345 0.3436
0.3794 0.3792
0.414 0.4106
0.4484 0.4466
0.4828 0.4822
0.5174 0.5202
0.5518 0.5528
0.5862 0.5874
0.6208 0.6242
0.6552 0.6612
0.6898 0.6988
0.7242 0.7324
0.7586 0.7678
0.7932 0.7982
0.8276 0.8274
0.862 0.861
0.8966 0.8986
0.931 0.9346
0.9656 0.9642
1 1
};
\addlegendentry{SGD T=10}
\addplot [semithick, mediumpurple148103189, mark=triangle*, mark size=2, mark repeat=10, mark options={solid,rotate=180}]
table {%
0.0002 0.0006
0.0346 0.0372
0.0692 0.092
0.1036 0.1322
0.1382 0.1736
0.1726 0.2126
0.207 0.257
0.2416 0.2908
0.276 0.3346
0.3104 0.3794
0.345 0.421
0.3794 0.4612
0.414 0.5028
0.4484 0.5378
0.4828 0.5818
0.5174 0.6212
0.5518 0.6476
0.5862 0.68
0.6208 0.7122
0.6552 0.7422
0.6898 0.7834
0.7242 0.813
0.7586 0.8464
0.7932 0.8722
0.8276 0.8986
0.862 0.925
0.8966 0.9476
0.931 0.972
0.9656 0.987
1 1
};
\addlegendentry{SGD T=25}
\addplot [semithick, sienna1408675, mark=triangle*, mark size=2, mark repeat=10, mark options={solid}]
table {%
0.0002 0.0008
0.0346 0.0698
0.0692 0.1486
0.1036 0.2222
0.1382 0.2796
0.1726 0.3398
0.207 0.39
0.2416 0.4426
0.276 0.4914
0.3104 0.5388
0.345 0.5866
0.3794 0.6326
0.414 0.6788
0.4484 0.723
0.4828 0.757
0.5174 0.783
0.5518 0.813
0.5862 0.8386
0.6208 0.8652
0.6552 0.8862
0.6898 0.9092
0.7242 0.925
0.7586 0.944
0.7932 0.9604
0.8276 0.9698
0.862 0.9814
0.8966 0.9886
0.931 0.9936
0.9656 0.9988
1 1
};
\addlegendentry{SGD T=50}
\end{axis}

\end{tikzpicture}

%% file: figures/roc-simulated/gaussian/Scaled_Gaussian_Kronecker_GLRT_vs_Scaled_Gaussian_Kronecker_SGD.tex
\begin{tikzpicture}

\definecolor{crimson2143940}{RGB}{214,39,40}
\definecolor{darkorange25512714}{RGB}{255,127,14}
\definecolor{darkslategray38}{RGB}{38,38,38}
\definecolor{forestgreen4416044}{RGB}{44,160,44}
\definecolor{lavender234234242}{RGB}{234,234,242}
\definecolor{lightgray204}{RGB}{204,204,204}
\definecolor{mediumpurple148103189}{RGB}{148,103,189}
\definecolor{sienna1408675}{RGB}{140,86,75}
\definecolor{steelblue31119180}{RGB}{31,119,180}

\begin{axis}[
axis background/.style={fill=lavender234234242},
axis line style={white},
height=5cm,
legend cell align={left},
legend style={
  fill opacity=0.8,
  draw opacity=1,
  text opacity=1,
  at={(1.05,0.5)},
  anchor=west,
  draw=lightgray204,
  fill=lavender234234242
},
tick align=outside,
width=0.9\columnwidth,
x grid style={white},
xlabel=\textcolor{darkslategray38}{\(\displaystyle \mathrm{P}_{\mathrm{FA}}\)},
xmajorgrids,
xmajorticks=true,
xmin=-0.04979, xmax=1.04999,
xtick style={color=darkslategray38},
y grid style={white},
ymajorgrids,
ymajorticks=false,
ymin=-0.05, ymax=1.05,
ytick style={color=darkslategray38},
title={Scaled-Gaussian Kronecker}
]
\addplot [semithick, steelblue31119180, mark=*, mark size=2, mark repeat=10, mark options={solid}]
table {%
0.0002 0.9932
0.0346 1
0.0692 1
0.1036 1
0.1382 1
0.1726 1
0.207 1
0.2416 1
0.276 1
0.3104 1
0.345 1
0.3794 1
0.414 1
0.4484 1
0.4828 1
0.5174 1
0.5518 1
0.5862 1
0.6208 1
0.6552 1
0.6898 1
0.7242 1
0.7586 1
0.7932 1
0.8276 1
0.862 1
0.8966 1
0.931 1
0.9656 1
1 1
};
\addlegendentry{GLRT T=50}
\addplot [semithick, darkorange25512714, mark=x, mark size=2, mark repeat=10, mark options={solid}]
table {%
0.0002 0
0.0346 0.0432
0.0692 0.0804
0.1036 0.1194
0.1382 0.158
0.1726 0.1908
0.207 0.2316
0.2416 0.267
0.276 0.3028
0.3104 0.3428
0.345 0.3756
0.3794 0.4104
0.414 0.4446
0.4484 0.477
0.4828 0.5128
0.5174 0.547
0.5518 0.582
0.5862 0.6114
0.6208 0.6418
0.6552 0.669
0.6898 0.708
0.7242 0.7446
0.7586 0.7834
0.7932 0.8132
0.8276 0.8452
0.862 0.8714
0.8966 0.9054
0.931 0.935
0.9656 0.9678
1 1
};
\addlegendentry{SGD T=2}
\addplot [semithick, forestgreen4416044, mark=square*, mark size=2, mark repeat=10, mark options={solid}]
table {%
0.0002 0
0.0346 0.0448
0.0692 0.088
0.1036 0.1248
0.1382 0.1654
0.1726 0.2012
0.207 0.2414
0.2416 0.2772
0.276 0.3212
0.3104 0.3564
0.345 0.3942
0.3794 0.4286
0.414 0.4662
0.4484 0.5106
0.4828 0.5424
0.5174 0.5692
0.5518 0.5982
0.5862 0.6334
0.6208 0.6694
0.6552 0.7018
0.6898 0.7364
0.7242 0.7732
0.7586 0.8034
0.7932 0.8376
0.8276 0.8686
0.862 0.8956
0.8966 0.921
0.931 0.9492
0.9656 0.9746
1 1
};
\addlegendentry{SGD T=5}
\addplot [semithick, crimson2143940, mark=diamond*, mark size=2, mark repeat=10, mark options={solid}]
table {%
0.0002 0
0.0346 0.051
0.0692 0.0984
0.1036 0.1388
0.1382 0.187
0.1726 0.2236
0.207 0.2606
0.2416 0.3066
0.276 0.345
0.3104 0.3888
0.345 0.4266
0.3794 0.4718
0.414 0.512
0.4484 0.5488
0.4828 0.5826
0.5174 0.6148
0.5518 0.6478
0.5862 0.6816
0.6208 0.7124
0.6552 0.7458
0.6898 0.7824
0.7242 0.8104
0.7586 0.837
0.7932 0.867
0.8276 0.8976
0.862 0.9218
0.8966 0.9484
0.931 0.9656
0.9656 0.9872
1 1
};
\addlegendentry{SGD T=10}
\addplot [semithick, mediumpurple148103189, mark=triangle*, mark size=2, mark repeat=10, mark options={solid,rotate=180}]
table {%
0.0002 0
0.0346 0.0836
0.0692 0.1542
0.1036 0.2142
0.1382 0.2728
0.1726 0.322
0.207 0.378
0.2416 0.433
0.276 0.484
0.3104 0.5278
0.345 0.5788
0.3794 0.6284
0.414 0.6724
0.4484 0.71
0.4828 0.7424
0.5174 0.7778
0.5518 0.8122
0.5862 0.8458
0.6208 0.8708
0.6552 0.8984
0.6898 0.922
0.7242 0.9402
0.7586 0.9604
0.7932 0.9758
0.8276 0.9882
0.862 0.9934
0.8966 0.997
0.931 0.9994
0.9656 1
1 1
};
\addlegendentry{SGD T=25}
\addplot [semithick, sienna1408675, mark=triangle*, mark size=2, mark repeat=10, mark options={solid}]
table {%
0.0002 0.0004
0.0346 0.193
0.0692 0.336
0.1036 0.449
0.1382 0.5742
0.1726 0.656
0.207 0.7262
0.2416 0.801
0.276 0.8564
0.3104 0.9008
0.345 0.9344
0.3794 0.9568
0.414 0.9706
0.4484 0.9816
0.4828 0.9884
0.5174 0.9938
0.5518 0.9972
0.5862 0.9982
0.6208 0.9996
0.6552 1
0.6898 1
0.7242 1
0.7586 1
0.7932 1
0.8276 1
0.862 1
0.8966 1
0.931 1
0.9656 1
1 1
};
\addlegendentry{SGD T=50}
\end{axis}

\end{tikzpicture}

%% file: figures/roc-simulated/pseudo-gaussian/Scaled_Gaussian_GLRT_vs_Scaled_Gaussian_SGD.tex
\begin{tikzpicture}

\definecolor{crimson2143940}{RGB}{214,39,40}
\definecolor{darkorange25512714}{RGB}{255,127,14}
\definecolor{darkslategray38}{RGB}{38,38,38}
\definecolor{forestgreen4416044}{RGB}{44,160,44}
\definecolor{lavender234234242}{RGB}{234,234,242}
\definecolor{lightgray204}{RGB}{204,204,204}
\definecolor{mediumpurple148103189}{RGB}{148,103,189}
\definecolor{sienna1408675}{RGB}{140,86,75}
\definecolor{steelblue31119180}{RGB}{31,119,180}

\begin{axis}[
axis background/.style={fill=lavender234234242},
axis line style={white},
height=5cm,
every axis legend/.code={\let\addlegendentry\relax},
tick align=outside,
width=1.01\columnwidth,
x grid style={white},
xlabel=\textcolor{darkslategray38}{\(\displaystyle \mathrm{P}_{\mathrm{FA}}\)},
xmajorgrids,
xmajorticks=true,
xmin=-0.04979, xmax=1.04999,
xtick style={color=darkslategray38},
y grid style={white},
ylabel=\textcolor{darkslategray38}{\(\displaystyle \mathrm{P}_{\mathrm{D}}\)},
ymajorgrids,
ymajorticks=true,
ymin=-0.05, ymax=1.05,
ytick style={color=darkslategray38},
title={Scaled-Gaussian}
]
\addplot [semithick, steelblue31119180, mark=*, mark size=2, mark repeat=10, mark options={solid}]
table {%
0.0002 0.064
0.0102 0.3072
0.0204 0.392
0.0304 0.4586
0.0406 0.5172
0.0506 0.5566
0.0608 0.5936
0.0708 0.6236
0.081 0.6476
0.091 0.6674
0.1012 0.6906
0.1112 0.7092
0.1214 0.7262
0.1314 0.7408
0.1416 0.7582
0.1516 0.7734
0.1618 0.7832
0.1718 0.795
0.182 0.808
0.192 0.821
0.2022 0.8306
0.2122 0.8408
0.2224 0.8486
0.2324 0.856
0.2426 0.8654
0.2526 0.8718
0.2628 0.8764
0.2728 0.8834
0.283 0.8874
0.293 0.892
0.3032 0.8974
0.3132 0.9022
0.3234 0.9082
0.3334 0.9114
0.3436 0.9162
0.3536 0.9194
0.3638 0.9234
0.3738 0.9266
0.384 0.9292
0.394 0.9326
0.4042 0.9364
0.4142 0.941
0.4244 0.944
0.4344 0.9462
0.4446 0.9502
0.4546 0.9518
0.4648 0.9546
0.4748 0.9556
0.485 0.9582
0.495 0.9606
0.5052 0.9622
0.5152 0.9646
0.5254 0.9654
0.5354 0.9672
0.5456 0.97
0.5556 0.9724
0.5658 0.9748
0.5758 0.9758
0.586 0.977
0.596 0.979
0.6062 0.9798
0.6162 0.9808
0.6264 0.982
0.6364 0.9834
0.6466 0.9852
0.6566 0.9862
0.6668 0.9876
0.6768 0.9878
0.687 0.9882
0.697 0.9886
0.7072 0.9902
0.7172 0.991
0.7274 0.9918
0.7374 0.9926
0.7476 0.994
0.7576 0.9944
0.7678 0.9948
0.7778 0.995
0.788 0.9954
0.798 0.9964
0.8082 0.9966
0.8182 0.997
0.8284 0.9974
0.8384 0.9976
0.8486 0.9978
0.8586 0.9986
0.8688 0.9988
0.8788 0.9994
0.889 0.9996
0.899 0.9996
0.9092 0.9996
0.9192 0.9998
0.9294 0.9998
0.9394 0.9998
0.9496 0.9998
0.9596 0.9998
0.9698 0.9998
0.9798 0.9998
0.99 1
1 1
};
\addlegendentry{GLRT T=50}
\addplot [semithick, darkorange25512714, mark=x, mark size=2, mark repeat=10, mark options={solid}]
table {%
0.0002 0
0.0102 0.0094
0.0204 0.0194
0.0304 0.0258
0.0406 0.0332
0.0506 0.0406
0.0608 0.0524
0.0708 0.0578
0.081 0.0668
0.091 0.0766
0.1012 0.0848
0.1112 0.093
0.1214 0.1004
0.1314 0.1118
0.1416 0.1194
0.1516 0.1254
0.1618 0.1374
0.1718 0.1446
0.182 0.1522
0.192 0.1624
0.2022 0.169
0.2122 0.1784
0.2224 0.1886
0.2324 0.196
0.2426 0.206
0.2526 0.2152
0.2628 0.222
0.2728 0.2296
0.283 0.2384
0.293 0.2494
0.3032 0.2598
0.3132 0.2658
0.3234 0.2726
0.3334 0.2816
0.3436 0.2892
0.3536 0.2986
0.3638 0.309
0.3738 0.3154
0.384 0.3242
0.394 0.33
0.4042 0.3398
0.4142 0.3492
0.4244 0.3594
0.4344 0.3682
0.4446 0.3828
0.4546 0.3922
0.4648 0.402
0.4748 0.4106
0.485 0.419
0.495 0.4274
0.5052 0.4376
0.5152 0.4466
0.5254 0.4582
0.5354 0.467
0.5456 0.4792
0.5556 0.4862
0.5658 0.502
0.5758 0.511
0.586 0.5206
0.596 0.5308
0.6062 0.542
0.6162 0.5548
0.6264 0.567
0.6364 0.576
0.6466 0.5834
0.6566 0.5928
0.6668 0.6026
0.6768 0.6112
0.687 0.6244
0.697 0.6334
0.7072 0.645
0.7172 0.6554
0.7274 0.6676
0.7374 0.6782
0.7476 0.687
0.7576 0.698
0.7678 0.7068
0.7778 0.7224
0.788 0.735
0.798 0.7446
0.8082 0.753
0.8182 0.7664
0.8284 0.7802
0.8384 0.79
0.8486 0.8032
0.8586 0.8146
0.8688 0.83
0.8788 0.8416
0.889 0.8518
0.899 0.861
0.9092 0.8762
0.9192 0.8848
0.9294 0.8982
0.9394 0.9092
0.9496 0.9222
0.9596 0.934
0.9698 0.9454
0.9798 0.961
0.99 0.982
1 0.9992
};
\addlegendentry{SGD T=2}
\addplot [semithick, forestgreen4416044, mark=square*, mark size=2, mark repeat=10, mark options={solid}]
table {%
0.0002 0.0002
0.0102 0.0076
0.0204 0.0178
0.0304 0.026
0.0406 0.036
0.0506 0.0444
0.0608 0.0538
0.0708 0.0622
0.081 0.0712
0.091 0.0812
0.1012 0.0934
0.1112 0.1036
0.1214 0.1122
0.1314 0.1248
0.1416 0.1348
0.1516 0.1432
0.1618 0.1496
0.1718 0.1572
0.182 0.17
0.192 0.1778
0.2022 0.1888
0.2122 0.1986
0.2224 0.2046
0.2324 0.2142
0.2426 0.2212
0.2526 0.2332
0.2628 0.245
0.2728 0.2578
0.283 0.2662
0.293 0.275
0.3032 0.283
0.3132 0.296
0.3234 0.3018
0.3334 0.308
0.3436 0.3162
0.3536 0.325
0.3638 0.3348
0.3738 0.3452
0.384 0.3546
0.394 0.362
0.4042 0.3724
0.4142 0.3804
0.4244 0.3906
0.4344 0.399
0.4446 0.4072
0.4546 0.4134
0.4648 0.4226
0.4748 0.4324
0.485 0.4452
0.495 0.4552
0.5052 0.4686
0.5152 0.4782
0.5254 0.4902
0.5354 0.5002
0.5456 0.5104
0.5556 0.518
0.5658 0.528
0.5758 0.5392
0.586 0.5516
0.596 0.5626
0.6062 0.575
0.6162 0.5864
0.6264 0.5936
0.6364 0.6064
0.6466 0.618
0.6566 0.6298
0.6668 0.6366
0.6768 0.6468
0.687 0.6544
0.697 0.6642
0.7072 0.674
0.7172 0.6832
0.7274 0.6924
0.7374 0.7018
0.7476 0.713
0.7576 0.7248
0.7678 0.7378
0.7778 0.7514
0.788 0.7668
0.798 0.7774
0.8082 0.7882
0.8182 0.7962
0.8284 0.809
0.8384 0.8228
0.8486 0.8306
0.8586 0.8418
0.8688 0.8542
0.8788 0.8658
0.889 0.8752
0.899 0.8878
0.9092 0.8968
0.9192 0.9074
0.9294 0.9174
0.9394 0.9294
0.9496 0.9384
0.9596 0.951
0.9698 0.9638
0.9798 0.975
0.99 0.9864
1 0.9996
};
\addlegendentry{SGD T=5}
\addplot [semithick, crimson2143940, mark=diamond*, mark size=2, mark repeat=10, mark options={solid}]
table {%
0.0002 0.0002
0.0102 0.0094
0.0204 0.0212
0.0304 0.0306
0.0406 0.0428
0.0506 0.0534
0.0608 0.0676
0.0708 0.0752
0.081 0.0876
0.091 0.0972
0.1012 0.107
0.1112 0.1162
0.1214 0.1246
0.1314 0.1336
0.1416 0.1458
0.1516 0.1516
0.1618 0.1632
0.1718 0.1712
0.182 0.1824
0.192 0.1906
0.2022 0.203
0.2122 0.2154
0.2224 0.2306
0.2324 0.2394
0.2426 0.2506
0.2526 0.259
0.2628 0.268
0.2728 0.2742
0.283 0.2868
0.293 0.2972
0.3032 0.3072
0.3132 0.3154
0.3234 0.3294
0.3334 0.3394
0.3436 0.35
0.3536 0.3592
0.3638 0.3654
0.3738 0.3732
0.384 0.3852
0.394 0.3974
0.4042 0.4082
0.4142 0.4186
0.4244 0.428
0.4344 0.4374
0.4446 0.4458
0.4546 0.4524
0.4648 0.4616
0.4748 0.4732
0.485 0.4838
0.495 0.4978
0.5052 0.5068
0.5152 0.513
0.5254 0.521
0.5354 0.5322
0.5456 0.5416
0.5556 0.5534
0.5658 0.568
0.5758 0.5768
0.586 0.5916
0.596 0.6032
0.6062 0.6154
0.6162 0.6236
0.6264 0.6366
0.6364 0.6504
0.6466 0.6578
0.6566 0.6698
0.6668 0.6784
0.6768 0.6866
0.687 0.6994
0.697 0.7112
0.7072 0.7246
0.7172 0.7344
0.7274 0.749
0.7374 0.7598
0.7476 0.7722
0.7576 0.7804
0.7678 0.7902
0.7778 0.7984
0.788 0.8066
0.798 0.8148
0.8082 0.8246
0.8182 0.833
0.8284 0.8434
0.8384 0.855
0.8486 0.8614
0.8586 0.871
0.8688 0.8792
0.8788 0.8888
0.889 0.9034
0.899 0.9118
0.9092 0.9246
0.9192 0.9338
0.9294 0.9424
0.9394 0.9492
0.9496 0.9596
0.9596 0.9672
0.9698 0.977
0.9798 0.9858
0.99 0.9932
1 1
};
\addlegendentry{SGD T=10}
\addplot [semithick, mediumpurple148103189, mark=triangle*, mark size=2, mark repeat=10, mark options={solid,rotate=180}]
table {%
0.0002 0.0002
0.0102 0.01
0.0204 0.0242
0.0304 0.0396
0.0406 0.0528
0.0506 0.0702
0.0608 0.0872
0.0708 0.1012
0.081 0.1152
0.091 0.1292
0.1012 0.1448
0.1112 0.1582
0.1214 0.1688
0.1314 0.1814
0.1416 0.196
0.1516 0.208
0.1618 0.2214
0.1718 0.2364
0.182 0.2464
0.192 0.26
0.2022 0.269
0.2122 0.2796
0.2224 0.2912
0.2324 0.3042
0.2426 0.3152
0.2526 0.3276
0.2628 0.3422
0.2728 0.3532
0.283 0.3654
0.293 0.3792
0.3032 0.3936
0.3132 0.4052
0.3234 0.4152
0.3334 0.4276
0.3436 0.4382
0.3536 0.4496
0.3638 0.4654
0.3738 0.4756
0.384 0.4886
0.394 0.4996
0.4042 0.5116
0.4142 0.5238
0.4244 0.5396
0.4344 0.5504
0.4446 0.5566
0.4546 0.5682
0.4648 0.5786
0.4748 0.5888
0.485 0.5972
0.495 0.6058
0.5052 0.6174
0.5152 0.63
0.5254 0.6374
0.5354 0.6506
0.5456 0.6656
0.5556 0.6742
0.5658 0.6856
0.5758 0.6958
0.586 0.7066
0.596 0.7132
0.6062 0.7224
0.6162 0.7366
0.6264 0.744
0.6364 0.7568
0.6466 0.7668
0.6566 0.7756
0.6668 0.782
0.6768 0.7886
0.687 0.7974
0.697 0.8074
0.7072 0.817
0.7172 0.8246
0.7274 0.8314
0.7374 0.8412
0.7476 0.8506
0.7576 0.8606
0.7678 0.8666
0.7778 0.877
0.788 0.886
0.798 0.8928
0.8082 0.9004
0.8182 0.9048
0.8284 0.9114
0.8384 0.92
0.8486 0.926
0.8586 0.9314
0.8688 0.9374
0.8788 0.9438
0.889 0.9476
0.899 0.9566
0.9092 0.962
0.9192 0.9684
0.9294 0.9742
0.9394 0.978
0.9496 0.9832
0.9596 0.987
0.9698 0.9922
0.9798 0.996
0.99 0.9978
1 1
};
\addlegendentry{SGD T=25}
\addplot [semithick, sienna1408675, mark=triangle*, mark size=2, mark repeat=10, mark options={solid}]
table {%
0.0002 0.0002
0.0102 0.0206
0.0204 0.0462
0.0304 0.0712
0.0406 0.1
0.0506 0.128
0.0608 0.1428
0.0708 0.177
0.081 0.2002
0.091 0.2248
0.1012 0.2452
0.1112 0.263
0.1214 0.2838
0.1314 0.3032
0.1416 0.3258
0.1516 0.3402
0.1618 0.36
0.1718 0.3788
0.182 0.3964
0.192 0.4136
0.2022 0.4434
0.2122 0.4518
0.2224 0.468
0.2324 0.4818
0.2426 0.4994
0.2526 0.5146
0.2628 0.525
0.2728 0.5416
0.283 0.5562
0.293 0.5732
0.3032 0.5894
0.3132 0.605
0.3234 0.62
0.3334 0.6318
0.3436 0.6446
0.3536 0.6572
0.3638 0.6658
0.3738 0.6742
0.384 0.6884
0.394 0.7008
0.4042 0.7094
0.4142 0.7216
0.4244 0.7308
0.4344 0.7414
0.4446 0.7538
0.4546 0.7612
0.4648 0.7714
0.4748 0.7776
0.485 0.7874
0.495 0.796
0.5052 0.8058
0.5152 0.8148
0.5254 0.823
0.5354 0.8294
0.5456 0.8394
0.5556 0.8456
0.5658 0.8556
0.5758 0.8606
0.586 0.8682
0.596 0.8748
0.6062 0.881
0.6162 0.8898
0.6264 0.8964
0.6364 0.9026
0.6466 0.9104
0.6566 0.9156
0.6668 0.9198
0.6768 0.9252
0.687 0.9308
0.697 0.935
0.7072 0.9394
0.7172 0.9416
0.7274 0.9468
0.7374 0.9504
0.7476 0.9546
0.7576 0.96
0.7678 0.9634
0.7778 0.9672
0.788 0.971
0.798 0.9744
0.8082 0.9764
0.8182 0.9784
0.8284 0.9806
0.8384 0.9828
0.8486 0.9848
0.8586 0.9864
0.8688 0.988
0.8788 0.989
0.889 0.9912
0.899 0.9922
0.9092 0.9938
0.9192 0.9948
0.9294 0.9956
0.9394 0.9968
0.9496 0.997
0.9596 0.9982
0.9698 0.9996
0.9798 1
0.99 1
1 1
};
\addlegendentry{SGD T=50}
\end{axis}

\end{tikzpicture}

%% file: figures/roc-simulated/pseudo-gaussian/Scaled_Gaussian_Kronecker_GLRT_vs_Scaled_Gaussian_Kronecker_SGD.tex
\begin{tikzpicture}

\definecolor{crimson2143940}{RGB}{214,39,40}
\definecolor{darkorange25512714}{RGB}{255,127,14}
\definecolor{darkslategray38}{RGB}{38,38,38}
\definecolor{forestgreen4416044}{RGB}{44,160,44}
\definecolor{lavender234234242}{RGB}{234,234,242}
\definecolor{lightgray204}{RGB}{204,204,204}
\definecolor{mediumpurple148103189}{RGB}{148,103,189}
\definecolor{sienna1408675}{RGB}{140,86,75}
\definecolor{steelblue31119180}{RGB}{31,119,180}

\begin{axis}[
axis background/.style={fill=lavender234234242},
axis line style={white},
height=5cm,
legend cell align={left},
legend style={
  fill opacity=0.8,
  draw opacity=1,
  text opacity=1,
  at={(1.05,0.5)},
  anchor=west,
  draw=lightgray204,
  fill=lavender234234242
},
tick align=outside,
width=0.9\columnwidth,
x grid style={white},
xlabel=\textcolor{darkslategray38}{\(\displaystyle \mathrm{P}_{\mathrm{FA}}\)},
xmajorgrids,
xmajorticks=true,
xmin=-0.04979, xmax=1.04999,
xtick style={color=darkslategray38},
y grid style={white},
ymajorgrids,
ymajorticks=false,
ymin=-0.05, ymax=1.05,
ytick style={color=darkslategray38},
title={Scaled-Gaussian Kronecker}
]
\addplot [semithick, steelblue31119180, mark=*, mark size=2, mark repeat=10, mark options={solid}]
table {%
0.0002 0.999
0.0102 1
0.0204 1
0.0304 1
0.0406 1
0.0506 1
0.0608 1
0.0708 1
0.081 1
0.091 1
0.1012 1
0.1112 1
0.1214 1
0.1314 1
0.1416 1
0.1516 1
0.1618 1
0.1718 1
0.182 1
0.192 1
0.2022 1
0.2122 1
0.2224 1
0.2324 1
0.2426 1
0.2526 1
0.2628 1
0.2728 1
0.283 1
0.293 1
0.3032 1
0.3132 1
0.3234 1
0.3334 1
0.3436 1
0.3536 1
0.3638 1
0.3738 1
0.384 1
0.394 1
0.4042 1
0.4142 1
0.4244 1
0.4344 1
0.4446 1
0.4546 1
0.4648 1
0.4748 1
0.485 1
0.495 1
0.5052 1
0.5152 1
0.5254 1
0.5354 1
0.5456 1
0.5556 1
0.5658 1
0.5758 1
0.586 1
0.596 1
0.6062 1
0.6162 1
0.6264 1
0.6364 1
0.6466 1
0.6566 1
0.6668 1
0.6768 1
0.687 1
0.697 1
0.7072 1
0.7172 1
0.7274 1
0.7374 1
0.7476 1
0.7576 1
0.7678 1
0.7778 1
0.788 1
0.798 1
0.8082 1
0.8182 1
0.8284 1
0.8384 1
0.8486 1
0.8586 1
0.8688 1
0.8788 1
0.889 1
0.899 1
0.9092 1
0.9192 1
0.9294 1
0.9394 1
0.9496 1
0.9596 1
0.9698 1
0.9798 1
0.99 1
1 1
};
\addlegendentry{GLRT T=50}
\addplot [semithick, darkorange25512714, mark=x, mark size=2, mark repeat=10, mark options={solid}]
table {%
0.0002 0.0004
0.0102 0.0194
0.0204 0.0288
0.0304 0.0412
0.0406 0.0602
0.0506 0.073
0.0608 0.0852
0.0708 0.0966
0.081 0.107
0.091 0.115
0.1012 0.1252
0.1112 0.1374
0.1214 0.146
0.1314 0.1538
0.1416 0.1648
0.1516 0.1754
0.1618 0.1866
0.1718 0.1994
0.182 0.21
0.192 0.2226
0.2022 0.2322
0.2122 0.2418
0.2224 0.2518
0.2324 0.263
0.2426 0.2722
0.2526 0.2812
0.2628 0.2958
0.2728 0.3068
0.283 0.3172
0.293 0.3274
0.3032 0.3354
0.3132 0.3456
0.3234 0.3556
0.3334 0.367
0.3436 0.3754
0.3536 0.3876
0.3638 0.3946
0.3738 0.4036
0.384 0.414
0.394 0.4272
0.4042 0.4378
0.4142 0.449
0.4244 0.4632
0.4344 0.4724
0.4446 0.48
0.4546 0.4888
0.4648 0.4986
0.4748 0.508
0.485 0.5178
0.495 0.5262
0.5052 0.537
0.5152 0.5448
0.5254 0.5544
0.5354 0.563
0.5456 0.5762
0.5556 0.5876
0.5658 0.5988
0.5758 0.6076
0.586 0.6186
0.596 0.628
0.6062 0.6376
0.6162 0.6522
0.6264 0.6614
0.6364 0.6708
0.6466 0.682
0.6566 0.692
0.6668 0.7022
0.6768 0.7112
0.687 0.7184
0.697 0.7298
0.7072 0.74
0.7172 0.7482
0.7274 0.7566
0.7374 0.7682
0.7476 0.7766
0.7576 0.7834
0.7678 0.7944
0.7778 0.8032
0.788 0.8122
0.798 0.8232
0.8082 0.8324
0.8182 0.8386
0.8284 0.8468
0.8384 0.8578
0.8486 0.8674
0.8586 0.8768
0.8688 0.8848
0.8788 0.8934
0.889 0.902
0.899 0.9128
0.9092 0.9206
0.9192 0.9312
0.9294 0.939
0.9394 0.9512
0.9496 0.9606
0.9596 0.9682
0.9698 0.9772
0.9798 0.9864
0.99 0.9916
1 1
};
\addlegendentry{SGD T=2}
\addplot [semithick, forestgreen4416044, mark=square*, mark size=2, mark repeat=10, mark options={solid}]
table {%
0.0002 0
0.0102 0.0156
0.0204 0.0324
0.0304 0.0466
0.0406 0.061
0.0506 0.0768
0.0608 0.0884
0.0708 0.0972
0.081 0.1102
0.091 0.1212
0.1012 0.1318
0.1112 0.1436
0.1214 0.156
0.1314 0.1678
0.1416 0.177
0.1516 0.1892
0.1618 0.2008
0.1718 0.2108
0.182 0.2236
0.192 0.2346
0.2022 0.2488
0.2122 0.2618
0.2224 0.2732
0.2324 0.285
0.2426 0.294
0.2526 0.3028
0.2628 0.3144
0.2728 0.324
0.283 0.337
0.293 0.3466
0.3032 0.3582
0.3132 0.3708
0.3234 0.3822
0.3334 0.3908
0.3436 0.4006
0.3536 0.4104
0.3638 0.416
0.3738 0.4278
0.384 0.4368
0.394 0.4474
0.4042 0.4604
0.4142 0.4706
0.4244 0.482
0.4344 0.4926
0.4446 0.4994
0.4546 0.5086
0.4648 0.5212
0.4748 0.5306
0.485 0.539
0.495 0.5486
0.5052 0.559
0.5152 0.572
0.5254 0.583
0.5354 0.5954
0.5456 0.607
0.5556 0.6178
0.5658 0.6282
0.5758 0.6356
0.586 0.649
0.596 0.661
0.6062 0.6726
0.6162 0.6802
0.6264 0.69
0.6364 0.7006
0.6466 0.7084
0.6566 0.7186
0.6668 0.7292
0.6768 0.7386
0.687 0.7506
0.697 0.7576
0.7072 0.765
0.7172 0.7708
0.7274 0.7812
0.7374 0.7912
0.7476 0.7982
0.7576 0.8112
0.7678 0.822
0.7778 0.8318
0.788 0.8394
0.798 0.848
0.8082 0.856
0.8182 0.8634
0.8284 0.874
0.8384 0.8836
0.8486 0.892
0.8586 0.8994
0.8688 0.905
0.8788 0.9138
0.889 0.9222
0.899 0.9294
0.9092 0.939
0.9192 0.9458
0.9294 0.9522
0.9394 0.9602
0.9496 0.9694
0.9596 0.9776
0.9698 0.9834
0.9798 0.9908
0.99 0.9946
1 1
};
\addlegendentry{SGD T=5}
\addplot [semithick, crimson2143940, mark=diamond*, mark size=2, mark repeat=10, mark options={solid}]
table {%
0.0002 0
0.0102 0.0182
0.0204 0.0348
0.0304 0.051
0.0406 0.0684
0.0506 0.0822
0.0608 0.099
0.0708 0.1094
0.081 0.1244
0.091 0.1358
0.1012 0.1512
0.1112 0.1646
0.1214 0.1756
0.1314 0.1878
0.1416 0.2012
0.1516 0.212
0.1618 0.2242
0.1718 0.2326
0.182 0.2468
0.192 0.2588
0.2022 0.272
0.2122 0.2824
0.2224 0.2936
0.2324 0.306
0.2426 0.3206
0.2526 0.3368
0.2628 0.3494
0.2728 0.3588
0.283 0.3676
0.293 0.377
0.3032 0.3852
0.3132 0.399
0.3234 0.4106
0.3334 0.4212
0.3436 0.4298
0.3536 0.442
0.3638 0.4532
0.3738 0.461
0.384 0.474
0.394 0.4874
0.4042 0.4948
0.4142 0.5058
0.4244 0.5188
0.4344 0.5344
0.4446 0.5438
0.4546 0.5544
0.4648 0.5638
0.4748 0.5752
0.485 0.5878
0.495 0.5976
0.5052 0.6096
0.5152 0.6212
0.5254 0.6312
0.5354 0.6418
0.5456 0.6544
0.5556 0.6658
0.5658 0.6776
0.5758 0.6854
0.586 0.696
0.596 0.704
0.6062 0.714
0.6162 0.7232
0.6264 0.7358
0.6364 0.7466
0.6466 0.7578
0.6566 0.7674
0.6668 0.7758
0.6768 0.7838
0.687 0.7928
0.697 0.805
0.7072 0.8142
0.7172 0.8218
0.7274 0.8278
0.7374 0.8352
0.7476 0.8416
0.7576 0.85
0.7678 0.863
0.7778 0.8718
0.788 0.8786
0.798 0.8874
0.8082 0.8936
0.8182 0.9018
0.8284 0.9118
0.8384 0.9188
0.8486 0.9266
0.8586 0.9346
0.8688 0.9406
0.8788 0.9468
0.889 0.9536
0.899 0.9612
0.9092 0.966
0.9192 0.9724
0.9294 0.976
0.9394 0.9812
0.9496 0.9854
0.9596 0.9886
0.9698 0.9934
0.9798 0.996
0.99 0.9986
1 1
};
\addlegendentry{SGD T=10}
\addplot [semithick, mediumpurple148103189, mark=triangle*, mark size=2, mark repeat=10, mark options={solid,rotate=180}]
table {%
0.0002 0.0002
0.0102 0.0286
0.0204 0.0584
0.0304 0.0834
0.0406 0.1088
0.0506 0.1286
0.0608 0.1532
0.0708 0.175
0.081 0.193
0.091 0.2122
0.1012 0.229
0.1112 0.2522
0.1214 0.2726
0.1314 0.2862
0.1416 0.3016
0.1516 0.3206
0.1618 0.3356
0.1718 0.3512
0.182 0.3686
0.192 0.3832
0.2022 0.4008
0.2122 0.4208
0.2224 0.4356
0.2324 0.452
0.2426 0.4682
0.2526 0.4806
0.2628 0.4944
0.2728 0.503
0.283 0.5182
0.293 0.5326
0.3032 0.548
0.3132 0.5654
0.3234 0.5832
0.3334 0.5962
0.3436 0.6122
0.3536 0.6294
0.3638 0.6468
0.3738 0.6588
0.384 0.67
0.394 0.6802
0.4042 0.6932
0.4142 0.7038
0.4244 0.7164
0.4344 0.726
0.4446 0.7336
0.4546 0.7462
0.4648 0.7566
0.4748 0.7672
0.485 0.7772
0.495 0.787
0.5052 0.7988
0.5152 0.8122
0.5254 0.8212
0.5354 0.8342
0.5456 0.8442
0.5556 0.8542
0.5658 0.8606
0.5758 0.8718
0.586 0.8802
0.596 0.8856
0.6062 0.897
0.6162 0.9048
0.6264 0.911
0.6364 0.9174
0.6466 0.9222
0.6566 0.9292
0.6668 0.9358
0.6768 0.9402
0.687 0.9442
0.697 0.9482
0.7072 0.9546
0.7172 0.9618
0.7274 0.9656
0.7374 0.9682
0.7476 0.971
0.7576 0.9738
0.7678 0.977
0.7778 0.98
0.788 0.982
0.798 0.9838
0.8082 0.9858
0.8182 0.988
0.8284 0.99
0.8384 0.9926
0.8486 0.9934
0.8586 0.9952
0.8688 0.9958
0.8788 0.997
0.889 0.9978
0.899 0.9984
0.9092 0.9994
0.9192 0.9998
0.9294 0.9998
0.9394 1
0.9496 1
0.9596 1
0.9698 1
0.9798 1
0.99 1
1 1
};
\addlegendentry{SGD T=25}
\addplot [semithick, sienna1408675, mark=triangle*, mark size=2, mark repeat=10, mark options={solid}]
table {%
0.0002 0.0002
0.0102 0.0748
0.0204 0.1292
0.0304 0.2028
0.0406 0.2544
0.0506 0.3018
0.0608 0.3378
0.0708 0.3896
0.081 0.4378
0.091 0.4832
0.1012 0.5308
0.1112 0.5692
0.1214 0.6032
0.1314 0.6346
0.1416 0.6718
0.1516 0.7008
0.1618 0.7312
0.1718 0.753
0.182 0.7794
0.192 0.7968
0.2022 0.8164
0.2122 0.83
0.2224 0.8462
0.2324 0.867
0.2426 0.8818
0.2526 0.894
0.2628 0.9038
0.2728 0.9112
0.283 0.9198
0.293 0.9292
0.3032 0.9392
0.3132 0.9478
0.3234 0.9544
0.3334 0.9626
0.3436 0.9672
0.3536 0.9722
0.3638 0.9778
0.3738 0.9812
0.384 0.9832
0.394 0.9874
0.4042 0.9886
0.4142 0.9902
0.4244 0.9914
0.4344 0.9928
0.4446 0.9938
0.4546 0.9948
0.4648 0.9958
0.4748 0.9966
0.485 0.9974
0.495 0.9982
0.5052 0.9986
0.5152 0.999
0.5254 0.9994
0.5354 0.9998
0.5456 0.9998
0.5556 0.9998
0.5658 1
0.5758 1
0.586 1
0.596 1
0.6062 1
0.6162 1
0.6264 1
0.6364 1
0.6466 1
0.6566 1
0.6668 1
0.6768 1
0.687 1
0.697 1
0.7072 1
0.7172 1
0.7274 1
0.7374 1
0.7476 1
0.7576 1
0.7678 1
0.7778 1
0.788 1
0.798 1
0.8082 1
0.8182 1
0.8284 1
0.8384 1
0.8486 1
0.8586 1
0.8688 1
0.8788 1
0.889 1
0.899 1
0.9092 1
0.9192 1
0.9294 1
0.9394 1
0.9496 1
0.9596 1
0.9698 1
0.9798 1
0.99 1
1 1
};
\addlegendentry{SGD T=50}
\end{axis}

\end{tikzpicture}

%% file: figures/roc-simulated/non-gaussian/Scaled_Gaussian_GLRT_vs_Scaled_Gaussian_SGD.tex
\begin{tikzpicture}

\definecolor{crimson2143940}{RGB}{214,39,40}
\definecolor{darkorange25512714}{RGB}{255,127,14}
\definecolor{darkslategray38}{RGB}{38,38,38}
\definecolor{forestgreen4416044}{RGB}{44,160,44}
\definecolor{lavender234234242}{RGB}{234,234,242}
\definecolor{lightgray204}{RGB}{204,204,204}
\definecolor{mediumpurple148103189}{RGB}{148,103,189}
\definecolor{sienna1408675}{RGB}{140,86,75}
\definecolor{steelblue31119180}{RGB}{31,119,180}

\begin{axis}[
axis background/.style={fill=lavender234234242},
axis line style={white},
height=5cm,
every axis legend/.code={\let\addlegendentry\relax},
tick align=outside,
width=1.01\columnwidth,
x grid style={white},
xlabel=\textcolor{darkslategray38}{\(\displaystyle \mathrm{P}_{\mathrm{FA}}\)},
xmajorgrids,
xmajorticks=true,
xmin=-0.04979, xmax=1.04999,
xtick style={color=darkslategray38},
y grid style={white},
ylabel=\textcolor{darkslategray38}{\(\displaystyle \mathrm{P}_{\mathrm{D}}\)},
ymajorgrids,
ymajorticks=true,
ymin=-0.05, ymax=1.05,
ytick style={color=darkslategray38},
title={Scaled-Gaussian}
]
\addplot [semithick, steelblue31119180, mark=*, mark size=2, mark repeat=10, mark options={solid}]
table {%
0.0002 0.9998
0.0102 1
0.0204 1
0.0304 1
0.0406 1
0.0506 1
0.0608 1
0.0708 1
0.081 1
0.091 1
0.1012 1
0.1112 1
0.1214 1
0.1314 1
0.1416 1
0.1516 1
0.1618 1
0.1718 1
0.182 1
0.192 1
0.2022 1
0.2122 1
0.2224 1
0.2324 1
0.2426 1
0.2526 1
0.2628 1
0.2728 1
0.283 1
0.293 1
0.3032 1
0.3132 1
0.3234 1
0.3334 1
0.3436 1
0.3536 1
0.3638 1
0.3738 1
0.384 1
0.394 1
0.4042 1
0.4142 1
0.4244 1
0.4344 1
0.4446 1
0.4546 1
0.4648 1
0.4748 1
0.485 1
0.495 1
0.5052 1
0.5152 1
0.5254 1
0.5354 1
0.5456 1
0.5556 1
0.5658 1
0.5758 1
0.586 1
0.596 1
0.6062 1
0.6162 1
0.6264 1
0.6364 1
0.6466 1
0.6566 1
0.6668 1
0.6768 1
0.687 1
0.697 1
0.7072 1
0.7172 1
0.7274 1
0.7374 1
0.7476 1
0.7576 1
0.7678 1
0.7778 1
0.788 1
0.798 1
0.8082 1
0.8182 1
0.8284 1
0.8384 1
0.8486 1
0.8586 1
0.8688 1
0.8788 1
0.889 1
0.899 1
0.9092 1
0.9192 1
0.9294 1
0.9394 1
0.9496 1
0.9596 1
0.9698 1
0.9798 1
0.99 1
1 1
};
\addlegendentry{GLRT T=50}
\addplot [semithick, darkorange25512714, mark=x, mark size=2, mark repeat=10, mark options={solid}]
table {%
0.0002 0.0094
0.0102 0.1582
0.0204 0.2308
0.0304 0.2768
0.0406 0.3232
0.0506 0.352
0.0608 0.3826
0.0708 0.4082
0.081 0.4274
0.091 0.447
0.1012 0.4648
0.1112 0.481
0.1214 0.494
0.1314 0.5054
0.1416 0.5264
0.1516 0.5388
0.1618 0.556
0.1718 0.5698
0.182 0.5802
0.192 0.5938
0.2022 0.6026
0.2122 0.6158
0.2224 0.628
0.2324 0.6372
0.2426 0.6492
0.2526 0.6584
0.2628 0.6668
0.2728 0.6758
0.283 0.6828
0.293 0.6886
0.3032 0.6942
0.3132 0.7022
0.3234 0.7122
0.3334 0.722
0.3436 0.7308
0.3536 0.7388
0.3638 0.7442
0.3738 0.752
0.384 0.756
0.394 0.7612
0.4042 0.7682
0.4142 0.7764
0.4244 0.783
0.4344 0.788
0.4446 0.7938
0.4546 0.802
0.4648 0.807
0.4748 0.8152
0.485 0.8226
0.495 0.8288
0.5052 0.8338
0.5152 0.8406
0.5254 0.8478
0.5354 0.8534
0.5456 0.8606
0.5556 0.8644
0.5658 0.8666
0.5758 0.8716
0.586 0.8756
0.596 0.8792
0.6062 0.8834
0.6162 0.8862
0.6264 0.889
0.6364 0.894
0.6466 0.8986
0.6566 0.9026
0.6668 0.9068
0.6768 0.9114
0.687 0.917
0.697 0.92
0.7072 0.9226
0.7172 0.9262
0.7274 0.9294
0.7374 0.9332
0.7476 0.9352
0.7576 0.9396
0.7678 0.9436
0.7778 0.9466
0.788 0.9498
0.798 0.9532
0.8082 0.957
0.8182 0.9592
0.8284 0.9632
0.8384 0.9668
0.8486 0.9688
0.8586 0.9728
0.8688 0.9754
0.8788 0.978
0.889 0.981
0.899 0.9832
0.9092 0.9846
0.9192 0.986
0.9294 0.9868
0.9394 0.9894
0.9496 0.9916
0.9596 0.9944
0.9698 0.995
0.9798 0.9966
0.99 0.9986
1 1
};
\addlegendentry{SGD T=2}
\addplot [semithick, forestgreen4416044, mark=square*, mark size=2, mark repeat=10, mark options={solid}]
table {%
0.0002 0.0034
0.0102 0.1582
0.0204 0.2536
0.0304 0.3146
0.0406 0.36
0.0506 0.3884
0.0608 0.4234
0.0708 0.4534
0.081 0.4786
0.091 0.5028
0.1012 0.5204
0.1112 0.5406
0.1214 0.5652
0.1314 0.5852
0.1416 0.6006
0.1516 0.612
0.1618 0.6246
0.1718 0.6406
0.182 0.6542
0.192 0.6668
0.2022 0.6802
0.2122 0.6944
0.2224 0.704
0.2324 0.7148
0.2426 0.7266
0.2526 0.7366
0.2628 0.743
0.2728 0.7518
0.283 0.7604
0.293 0.7674
0.3032 0.7732
0.3132 0.7818
0.3234 0.7904
0.3334 0.7976
0.3436 0.8032
0.3536 0.8088
0.3638 0.815
0.3738 0.8238
0.384 0.8304
0.394 0.8366
0.4042 0.8434
0.4142 0.8494
0.4244 0.8536
0.4344 0.8594
0.4446 0.8636
0.4546 0.871
0.4648 0.8776
0.4748 0.8828
0.485 0.8858
0.495 0.889
0.5052 0.8916
0.5152 0.8946
0.5254 0.8992
0.5354 0.9032
0.5456 0.9068
0.5556 0.9108
0.5658 0.9138
0.5758 0.917
0.586 0.921
0.596 0.9268
0.6062 0.9284
0.6162 0.931
0.6264 0.9344
0.6364 0.9384
0.6466 0.9428
0.6566 0.9464
0.6668 0.9492
0.6768 0.951
0.687 0.9534
0.697 0.9556
0.7072 0.9582
0.7172 0.9598
0.7274 0.962
0.7374 0.9636
0.7476 0.9648
0.7576 0.9664
0.7678 0.9686
0.7778 0.9712
0.788 0.9732
0.798 0.9738
0.8082 0.9746
0.8182 0.9766
0.8284 0.9782
0.8384 0.9808
0.8486 0.9836
0.8586 0.9844
0.8688 0.9852
0.8788 0.9864
0.889 0.9876
0.899 0.989
0.9092 0.9902
0.9192 0.992
0.9294 0.9932
0.9394 0.9952
0.9496 0.9962
0.9596 0.9976
0.9698 0.999
0.9798 1
0.99 1
1 1
};
\addlegendentry{SGD T=5}
\addplot [semithick, crimson2143940, mark=diamond*, mark size=2, mark repeat=10, mark options={solid}]
table {%
0.0002 0.0036
0.0102 0.2364
0.0204 0.346
0.0304 0.4126
0.0406 0.4758
0.0506 0.5112
0.0608 0.5454
0.0708 0.579
0.081 0.6088
0.091 0.6318
0.1012 0.6556
0.1112 0.6856
0.1214 0.7002
0.1314 0.716
0.1416 0.7304
0.1516 0.7478
0.1618 0.758
0.1718 0.7702
0.182 0.7788
0.192 0.79
0.2022 0.8044
0.2122 0.8154
0.2224 0.8252
0.2324 0.8354
0.2426 0.8452
0.2526 0.8506
0.2628 0.8588
0.2728 0.867
0.283 0.8716
0.293 0.8786
0.3032 0.8854
0.3132 0.8894
0.3234 0.8958
0.3334 0.8996
0.3436 0.903
0.3536 0.9064
0.3638 0.9108
0.3738 0.9154
0.384 0.9202
0.394 0.9238
0.4042 0.9272
0.4142 0.931
0.4244 0.9346
0.4344 0.9382
0.4446 0.9414
0.4546 0.944
0.4648 0.9462
0.4748 0.9488
0.485 0.9518
0.495 0.9538
0.5052 0.9556
0.5152 0.9574
0.5254 0.9594
0.5354 0.9616
0.5456 0.9642
0.5556 0.9656
0.5658 0.967
0.5758 0.9696
0.586 0.971
0.596 0.9724
0.6062 0.974
0.6162 0.9766
0.6264 0.9784
0.6364 0.98
0.6466 0.9808
0.6566 0.983
0.6668 0.984
0.6768 0.9846
0.687 0.9864
0.697 0.9884
0.7072 0.9888
0.7172 0.9896
0.7274 0.9896
0.7374 0.9902
0.7476 0.9918
0.7576 0.9928
0.7678 0.9938
0.7778 0.9938
0.788 0.9938
0.798 0.9938
0.8082 0.9944
0.8182 0.995
0.8284 0.9954
0.8384 0.9962
0.8486 0.9968
0.8586 0.9974
0.8688 0.9974
0.8788 0.9978
0.889 0.9982
0.899 0.9988
0.9092 0.999
0.9192 0.9996
0.9294 0.9996
0.9394 0.9996
0.9496 0.9998
0.9596 0.9998
0.9698 0.9998
0.9798 1
0.99 1
1 1
};
\addlegendentry{SGD T=10}
\addplot [semithick, mediumpurple148103189, mark=triangle*, mark size=2, mark repeat=10, mark options={solid,rotate=180}]
table {%
0.0002 0.0266
0.0102 0.5178
0.0204 0.6562
0.0304 0.7368
0.0406 0.7832
0.0506 0.8246
0.0608 0.8494
0.0708 0.8726
0.081 0.889
0.091 0.9026
0.1012 0.9142
0.1112 0.9228
0.1214 0.9306
0.1314 0.9376
0.1416 0.9436
0.1516 0.9492
0.1618 0.9538
0.1718 0.957
0.182 0.9616
0.192 0.9646
0.2022 0.9662
0.2122 0.9692
0.2224 0.9728
0.2324 0.9756
0.2426 0.978
0.2526 0.9794
0.2628 0.9814
0.2728 0.9836
0.283 0.9848
0.293 0.9868
0.3032 0.9874
0.3132 0.9884
0.3234 0.9892
0.3334 0.9902
0.3436 0.9908
0.3536 0.9908
0.3638 0.9914
0.3738 0.9922
0.384 0.9928
0.394 0.9944
0.4042 0.9952
0.4142 0.9954
0.4244 0.9958
0.4344 0.996
0.4446 0.9964
0.4546 0.9966
0.4648 0.9968
0.4748 0.997
0.485 0.997
0.495 0.9976
0.5052 0.9976
0.5152 0.9978
0.5254 0.9982
0.5354 0.9982
0.5456 0.9984
0.5556 0.9986
0.5658 0.999
0.5758 0.999
0.586 0.999
0.596 0.9994
0.6062 0.9994
0.6162 0.9998
0.6264 0.9998
0.6364 0.9998
0.6466 0.9998
0.6566 0.9998
0.6668 0.9998
0.6768 0.9998
0.687 1
0.697 1
0.7072 1
0.7172 1
0.7274 1
0.7374 1
0.7476 1
0.7576 1
0.7678 1
0.7778 1
0.788 1
0.798 1
0.8082 1
0.8182 1
0.8284 1
0.8384 1
0.8486 1
0.8586 1
0.8688 1
0.8788 1
0.889 1
0.899 1
0.9092 1
0.9192 1
0.9294 1
0.9394 1
0.9496 1
0.9596 1
0.9698 1
0.9798 1
0.99 1
1 1
};
\addlegendentry{SGD T=25}
\addplot [semithick, sienna1408675, mark=triangle*, mark size=2, mark repeat=10, mark options={solid}]
table {%
0.0002 0.1248
0.0102 0.8158
0.0204 0.9078
0.0304 0.9426
0.0406 0.9646
0.0506 0.9738
0.0608 0.98
0.0708 0.9862
0.081 0.9884
0.091 0.9924
0.1012 0.993
0.1112 0.9942
0.1214 0.9954
0.1314 0.9956
0.1416 0.9962
0.1516 0.9966
0.1618 0.997
0.1718 0.9974
0.182 0.9976
0.192 0.9978
0.2022 0.9984
0.2122 0.9988
0.2224 0.999
0.2324 0.999
0.2426 0.999
0.2526 0.999
0.2628 0.9996
0.2728 0.9996
0.283 0.9996
0.293 0.9996
0.3032 0.9996
0.3132 0.9996
0.3234 0.9996
0.3334 0.9996
0.3436 0.9996
0.3536 0.9996
0.3638 0.9998
0.3738 0.9998
0.384 0.9998
0.394 0.9998
0.4042 1
0.4142 1
0.4244 1
0.4344 1
0.4446 1
0.4546 1
0.4648 1
0.4748 1
0.485 1
0.495 1
0.5052 1
0.5152 1
0.5254 1
0.5354 1
0.5456 1
0.5556 1
0.5658 1
0.5758 1
0.586 1
0.596 1
0.6062 1
0.6162 1
0.6264 1
0.6364 1
0.6466 1
0.6566 1
0.6668 1
0.6768 1
0.687 1
0.697 1
0.7072 1
0.7172 1
0.7274 1
0.7374 1
0.7476 1
0.7576 1
0.7678 1
0.7778 1
0.788 1
0.798 1
0.8082 1
0.8182 1
0.8284 1
0.8384 1
0.8486 1
0.8586 1
0.8688 1
0.8788 1
0.889 1
0.899 1
0.9092 1
0.9192 1
0.9294 1
0.9394 1
0.9496 1
0.9596 1
0.9698 1
0.9798 1
0.99 1
1 1
};
\addlegendentry{SGD T=50}
\end{axis}

\end{tikzpicture}

%% file: figures/roc-simulated/non-gaussian/Scaled_Gaussian_Kronecker_GLRT_vs_Scaled_Gaussian_Kronecker_SGD.tex
\begin{tikzpicture}

\definecolor{crimson2143940}{RGB}{214,39,40}
\definecolor{darkorange25512714}{RGB}{255,127,14}
\definecolor{darkslategray38}{RGB}{38,38,38}
\definecolor{forestgreen4416044}{RGB}{44,160,44}
\definecolor{lavender234234242}{RGB}{234,234,242}
\definecolor{lightgray204}{RGB}{204,204,204}
\definecolor{mediumpurple148103189}{RGB}{148,103,189}
\definecolor{sienna1408675}{RGB}{140,86,75}
\definecolor{steelblue31119180}{RGB}{31,119,180}

\begin{axis}[
axis background/.style={fill=lavender234234242},
axis line style={white},
height=5cm,
legend cell align={left},
legend style={
  fill opacity=0.8,
  draw opacity=1,
  text opacity=1,
  at={(1.05,0.5)},
  anchor=west,
  draw=lightgray204,
  fill=lavender234234242
},
tick align=outside,
width=0.9\columnwidth,
x grid style={white},
xlabel=\textcolor{darkslategray38}{\(\displaystyle \mathrm{P}_{\mathrm{FA}}\)},
xmajorgrids,
xmajorticks=true,
xmin=-0.04979, xmax=1.04999,
xtick style={color=darkslategray38},
y grid style={white},
ymajorgrids,
ymajorticks=false,
ymin=-0.05, ymax=1.05,
ytick style={color=darkslategray38},
title={Scaled-Gaussian Kronecker}
]
\addplot [semithick, steelblue31119180, mark=*, mark size=2, mark repeat=10, mark options={solid}]
table {%
0.0002 1
0.0102 1
0.0204 1
0.0304 1
0.0406 1
0.0506 1
0.0608 1
0.0708 1
0.081 1
0.091 1
0.1012 1
0.1112 1
0.1214 1
0.1314 1
0.1416 1
0.1516 1
0.1618 1
0.1718 1
0.182 1
0.192 1
0.2022 1
0.2122 1
0.2224 1
0.2324 1
0.2426 1
0.2526 1
0.2628 1
0.2728 1
0.283 1
0.293 1
0.3032 1
0.3132 1
0.3234 1
0.3334 1
0.3436 1
0.3536 1
0.3638 1
0.3738 1
0.384 1
0.394 1
0.4042 1
0.4142 1
0.4244 1
0.4344 1
0.4446 1
0.4546 1
0.4648 1
0.4748 1
0.485 1
0.495 1
0.5052 1
0.5152 1
0.5254 1
0.5354 1
0.5456 1
0.5556 1
0.5658 1
0.5758 1
0.586 1
0.596 1
0.6062 1
0.6162 1
0.6264 1
0.6364 1
0.6466 1
0.6566 1
0.6668 1
0.6768 1
0.687 1
0.697 1
0.7072 1
0.7172 1
0.7274 1
0.7374 1
0.7476 1
0.7576 1
0.7678 1
0.7778 1
0.788 1
0.798 1
0.8082 1
0.8182 1
0.8284 1
0.8384 1
0.8486 1
0.8586 1
0.8688 1
0.8788 1
0.889 1
0.899 1
0.9092 1
0.9192 1
0.9294 1
0.9394 1
0.9496 1
0.9596 1
0.9698 1
0.9798 1
0.99 1
1 1
};
\addlegendentry{GLRT T=50}
\addplot [semithick, darkorange25512714, mark=x, mark size=2, mark repeat=10, mark options={solid}]
table {%
0.0002 0.0426
0.0102 0.2326
0.0204 0.2956
0.0304 0.3368
0.0406 0.3682
0.0506 0.3998
0.0608 0.4226
0.0708 0.4462
0.081 0.4736
0.091 0.493
0.1012 0.5104
0.1112 0.5286
0.1214 0.5422
0.1314 0.5592
0.1416 0.5766
0.1516 0.5882
0.1618 0.6032
0.1718 0.618
0.182 0.6284
0.192 0.6428
0.2022 0.6524
0.2122 0.6606
0.2224 0.6708
0.2324 0.679
0.2426 0.6864
0.2526 0.6972
0.2628 0.7038
0.2728 0.7126
0.283 0.719
0.293 0.7246
0.3032 0.731
0.3132 0.7392
0.3234 0.7458
0.3334 0.7528
0.3436 0.7584
0.3536 0.7644
0.3638 0.7716
0.3738 0.778
0.384 0.7838
0.394 0.7914
0.4042 0.7976
0.4142 0.8058
0.4244 0.8098
0.4344 0.8152
0.4446 0.8222
0.4546 0.8296
0.4648 0.8328
0.4748 0.8398
0.485 0.8452
0.495 0.8492
0.5052 0.8532
0.5152 0.8594
0.5254 0.8646
0.5354 0.8696
0.5456 0.8746
0.5556 0.8784
0.5658 0.8826
0.5758 0.888
0.586 0.8922
0.596 0.8954
0.6062 0.899
0.6162 0.9046
0.6264 0.9086
0.6364 0.9138
0.6466 0.9168
0.6566 0.9214
0.6668 0.9248
0.6768 0.9288
0.687 0.932
0.697 0.9344
0.7072 0.9388
0.7172 0.9416
0.7274 0.9452
0.7374 0.9468
0.7476 0.949
0.7576 0.9522
0.7678 0.9538
0.7778 0.956
0.788 0.9576
0.798 0.9612
0.8082 0.963
0.8182 0.9658
0.8284 0.969
0.8384 0.973
0.8486 0.9742
0.8586 0.9768
0.8688 0.9788
0.8788 0.9802
0.889 0.9826
0.899 0.9834
0.9092 0.986
0.9192 0.9884
0.9294 0.9914
0.9394 0.9924
0.9496 0.9948
0.9596 0.9956
0.9698 0.9978
0.9798 0.9986
0.99 0.9996
1 1
};
\addlegendentry{SGD T=2}
\addplot [semithick, forestgreen4416044, mark=square*, mark size=2, mark repeat=10, mark options={solid}]
table {%
0.0002 0.039
0.0102 0.2526
0.0204 0.3406
0.0304 0.39
0.0406 0.4358
0.0506 0.4692
0.0608 0.4948
0.0708 0.521
0.081 0.544
0.091 0.563
0.1012 0.5802
0.1112 0.6004
0.1214 0.6228
0.1314 0.6402
0.1416 0.6534
0.1516 0.6692
0.1618 0.6844
0.1718 0.697
0.182 0.71
0.192 0.7234
0.2022 0.7402
0.2122 0.7534
0.2224 0.762
0.2324 0.7678
0.2426 0.7756
0.2526 0.7844
0.2628 0.7932
0.2728 0.8028
0.283 0.8108
0.293 0.8196
0.3032 0.8254
0.3132 0.8356
0.3234 0.8402
0.3334 0.8472
0.3436 0.8508
0.3536 0.8582
0.3638 0.8616
0.3738 0.868
0.384 0.872
0.394 0.8782
0.4042 0.8828
0.4142 0.8876
0.4244 0.892
0.4344 0.8962
0.4446 0.9022
0.4546 0.905
0.4648 0.91
0.4748 0.913
0.485 0.9162
0.495 0.9208
0.5052 0.9236
0.5152 0.9272
0.5254 0.9298
0.5354 0.9332
0.5456 0.9352
0.5556 0.9382
0.5658 0.9414
0.5758 0.9446
0.586 0.9488
0.596 0.9518
0.6062 0.953
0.6162 0.956
0.6264 0.957
0.6364 0.9598
0.6466 0.9606
0.6566 0.9634
0.6668 0.9658
0.6768 0.9682
0.687 0.969
0.697 0.9704
0.7072 0.9714
0.7172 0.9732
0.7274 0.975
0.7374 0.9768
0.7476 0.9786
0.7576 0.9796
0.7678 0.981
0.7778 0.9834
0.788 0.9844
0.798 0.987
0.8082 0.9884
0.8182 0.9902
0.8284 0.991
0.8384 0.9916
0.8486 0.9934
0.8586 0.9944
0.8688 0.9952
0.8788 0.996
0.889 0.9962
0.899 0.9974
0.9092 0.9984
0.9192 0.9984
0.9294 0.9988
0.9394 0.9992
0.9496 0.9994
0.9596 1
0.9698 1
0.9798 1
0.99 1
1 1
};
\addlegendentry{SGD T=5}
\addplot [semithick, crimson2143940, mark=diamond*, mark size=2, mark repeat=10, mark options={solid}]
table {%
0.0002 0.0494
0.0102 0.3622
0.0204 0.4816
0.0304 0.5522
0.0406 0.6046
0.0506 0.6386
0.0608 0.6762
0.0708 0.7022
0.081 0.7278
0.091 0.746
0.1012 0.7654
0.1112 0.7804
0.1214 0.7998
0.1314 0.8146
0.1416 0.8286
0.1516 0.8416
0.1618 0.8574
0.1718 0.8674
0.182 0.8738
0.192 0.8808
0.2022 0.8858
0.2122 0.893
0.2224 0.9018
0.2324 0.9074
0.2426 0.9134
0.2526 0.9198
0.2628 0.928
0.2728 0.9312
0.283 0.9356
0.293 0.9382
0.3032 0.941
0.3132 0.9446
0.3234 0.9482
0.3334 0.9516
0.3436 0.954
0.3536 0.956
0.3638 0.9586
0.3738 0.961
0.384 0.963
0.394 0.9664
0.4042 0.969
0.4142 0.9702
0.4244 0.9724
0.4344 0.9738
0.4446 0.9752
0.4546 0.9766
0.4648 0.9778
0.4748 0.9796
0.485 0.9814
0.495 0.9822
0.5052 0.9836
0.5152 0.9852
0.5254 0.986
0.5354 0.9876
0.5456 0.989
0.5556 0.9896
0.5658 0.9904
0.5758 0.9912
0.586 0.9924
0.596 0.9932
0.6062 0.994
0.6162 0.994
0.6264 0.9946
0.6364 0.9952
0.6466 0.9964
0.6566 0.9968
0.6668 0.9974
0.6768 0.9976
0.687 0.9978
0.697 0.9982
0.7072 0.9982
0.7172 0.9982
0.7274 0.9984
0.7374 0.9986
0.7476 0.9986
0.7576 0.9988
0.7678 0.9992
0.7778 0.9994
0.788 0.9996
0.798 0.9996
0.8082 0.9996
0.8182 0.9996
0.8284 0.9996
0.8384 0.9996
0.8486 0.9996
0.8586 0.9996
0.8688 0.9996
0.8788 0.9998
0.889 0.9998
0.899 1
0.9092 1
0.9192 1
0.9294 1
0.9394 1
0.9496 1
0.9596 1
0.9698 1
0.9798 1
0.99 1
1 1
};
\addlegendentry{SGD T=10}
\addplot [semithick, mediumpurple148103189, mark=triangle*, mark size=2, mark repeat=10, mark options={solid,rotate=180}]
table {%
0.0002 0.1306
0.0102 0.7122
0.0204 0.8348
0.0304 0.8842
0.0406 0.9168
0.0506 0.9324
0.0608 0.9446
0.0708 0.9552
0.081 0.965
0.091 0.9732
0.1012 0.9788
0.1112 0.9814
0.1214 0.9846
0.1314 0.9864
0.1416 0.9882
0.1516 0.9898
0.1618 0.9914
0.1718 0.9926
0.182 0.9928
0.192 0.994
0.2022 0.9944
0.2122 0.9956
0.2224 0.9956
0.2324 0.9956
0.2426 0.996
0.2526 0.9964
0.2628 0.997
0.2728 0.9974
0.283 0.9976
0.293 0.9978
0.3032 0.9988
0.3132 0.999
0.3234 0.999
0.3334 0.9992
0.3436 0.9992
0.3536 0.9992
0.3638 0.9994
0.3738 0.9994
0.384 0.9996
0.394 0.9998
0.4042 0.9998
0.4142 0.9998
0.4244 0.9998
0.4344 0.9998
0.4446 1
0.4546 1
0.4648 1
0.4748 1
0.485 1
0.495 1
0.5052 1
0.5152 1
0.5254 1
0.5354 1
0.5456 1
0.5556 1
0.5658 1
0.5758 1
0.586 1
0.596 1
0.6062 1
0.6162 1
0.6264 1
0.6364 1
0.6466 1
0.6566 1
0.6668 1
0.6768 1
0.687 1
0.697 1
0.7072 1
0.7172 1
0.7274 1
0.7374 1
0.7476 1
0.7576 1
0.7678 1
0.7778 1
0.788 1
0.798 1
0.8082 1
0.8182 1
0.8284 1
0.8384 1
0.8486 1
0.8586 1
0.8688 1
0.8788 1
0.889 1
0.899 1
0.9092 1
0.9192 1
0.9294 1
0.9394 1
0.9496 1
0.9596 1
0.9698 1
0.9798 1
0.99 1
1 1
};
\addlegendentry{SGD T=25}
\addplot [semithick, sienna1408675, mark=triangle*, mark size=2, mark repeat=10, mark options={solid}]
table {%
0.0002 0.2868
0.0102 0.935
0.0204 0.9766
0.0304 0.9898
0.0406 0.994
0.0506 0.9968
0.0608 0.9978
0.0708 0.9986
0.081 0.999
0.091 0.9992
0.1012 0.9992
0.1112 0.9996
0.1214 0.9998
0.1314 0.9998
0.1416 1
0.1516 1
0.1618 1
0.1718 1
0.182 1
0.192 1
0.2022 1
0.2122 1
0.2224 1
0.2324 1
0.2426 1
0.2526 1
0.2628 1
0.2728 1
0.283 1
0.293 1
0.3032 1
0.3132 1
0.3234 1
0.3334 1
0.3436 1
0.3536 1
0.3638 1
0.3738 1
0.384 1
0.394 1
0.4042 1
0.4142 1
0.4244 1
0.4344 1
0.4446 1
0.4546 1
0.4648 1
0.4748 1
0.485 1
0.495 1
0.5052 1
0.5152 1
0.5254 1
0.5354 1
0.5456 1
0.5556 1
0.5658 1
0.5758 1
0.586 1
0.596 1
0.6062 1
0.6162 1
0.6264 1
0.6364 1
0.6466 1
0.6566 1
0.6668 1
0.6768 1
0.687 1
0.697 1
0.7072 1
0.7172 1
0.7274 1
0.7374 1
0.7476 1
0.7576 1
0.7678 1
0.7778 1
0.788 1
0.798 1
0.8082 1
0.8182 1
0.8284 1
0.8384 1
0.8486 1
0.8586 1
0.8688 1
0.8788 1
0.889 1
0.899 1
0.9092 1
0.9192 1
0.9294 1
0.9394 1
0.9496 1
0.9596 1
0.9698 1
0.9798 1
0.99 1
1 1
};
\addlegendentry{SGD T=50}
\end{axis}

\end{tikzpicture}

%% file: figures/realdata/repeat-none/scaled-gaussian-glrt.tex
\begin{tikzpicture}

\definecolor{darkgray176}{RGB}{176,176,176}

\begin{axis}[
colorbar,
colorbar style={ylabel={}},
colormap={mymap}{[1pt]
  rgb(0pt)=(0,0,0.5);
  rgb(22pt)=(0,0,1);
  rgb(25pt)=(0,0,1);
  rgb(68pt)=(0,0.86,1);
  rgb(70pt)=(0,0.9,0.967741935483871);
  rgb(75pt)=(0.0806451612903226,1,0.887096774193548);
  rgb(128pt)=(0.935483870967742,1,0.0322580645161291);
  rgb(130pt)=(0.967741935483871,0.962962962962963,0);
  rgb(132pt)=(1,0.925925925925926,0);
  rgb(178pt)=(1,0.0740740740740741,0);
  rgb(182pt)=(0.909090909090909,0,0);
  rgb(200pt)=(0.5,0,0)
},
height=4cm,
point meta max=54698.8538975652,
point meta min=24422.8415353895,
tick align=outside,
tick pos=left,
width=\textwidth,
x grid style={darkgray176},
xmin=-0.5, xmax=193.5,
y dir=reverse,
y grid style={darkgray176},
ymin=-0.5, ymax=193.5,
ticks=none
]
\addplot graphics [includegraphics cmd=\pgfimage,xmin=-0.5, xmax=193.5, ymin=193.5, ymax=-0.5] {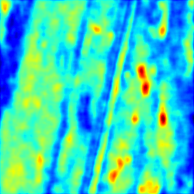};
\end{axis}

\end{tikzpicture}

%% file: figures/realdata/repeat-none/scaled-gaussian-sgd.tex
\begin{tikzpicture}

\definecolor{darkgray176}{RGB}{176,176,176}

\begin{axis}[
colorbar,
colorbar style={ylabel={}},
colormap={mymap}{[1pt]
  rgb(0pt)=(0,0,0.5);
  rgb(22pt)=(0,0,1);
  rgb(25pt)=(0,0,1);
  rgb(68pt)=(0,0.86,1);
  rgb(70pt)=(0,0.9,0.967741935483871);
  rgb(75pt)=(0.0806451612903226,1,0.887096774193548);
  rgb(128pt)=(0.935483870967742,1,0.0322580645161291);
  rgb(130pt)=(0.967741935483871,0.962962962962963,0);
  rgb(132pt)=(1,0.925925925925926,0);
  rgb(178pt)=(1,0.0740740740740741,0);
  rgb(182pt)=(0.909090909090909,0,0);
  rgb(200pt)=(0.5,0,0)
},
height=4cm,
point meta max=54698.8538975652,
point meta min=24422.8415353895,
tick align=outside,
tick pos=left,
width=\textwidth,
x grid style={darkgray176},
xmin=-0.5, xmax=193.5,
y dir=reverse,
y grid style={darkgray176},
ymin=-0.5, ymax=193.5,
ticks=none
]
\addplot graphics [includegraphics cmd=\pgfimage,xmin=-0.5, xmax=193.5, ymin=193.5, ymax=-0.5] {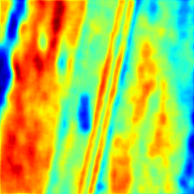};
\end{axis}

\end{tikzpicture}

%% file: figures/realdata/repeat-none/scaled-gaussian-kron-glrt.tex
\begin{tikzpicture}

\definecolor{darkgray176}{RGB}{176,176,176}

\begin{axis}[
colorbar,
colorbar style={ylabel={}},
colormap={mymap}{[1pt]
  rgb(0pt)=(0,0,0.5);
  rgb(22pt)=(0,0,1);
  rgb(25pt)=(0,0,1);
  rgb(68pt)=(0,0.86,1);
  rgb(70pt)=(0,0.9,0.967741935483871);
  rgb(75pt)=(0.0806451612903226,1,0.887096774193548);
  rgb(128pt)=(0.935483870967742,1,0.0322580645161291);
  rgb(130pt)=(0.967741935483871,0.962962962962963,0);
  rgb(132pt)=(1,0.925925925925926,0);
  rgb(178pt)=(1,0.0740740740740741,0);
  rgb(182pt)=(0.909090909090909,0,0);
  rgb(200pt)=(0.5,0,0)
},
height=4cm,
point meta max=54698.8538975652,
point meta min=24422.8415353895,
tick align=outside,
tick pos=left,
width=\textwidth,
x grid style={darkgray176},
xmin=-0.5, xmax=193.5,
y dir=reverse,
y grid style={darkgray176},
ymin=-0.5, ymax=193.5,
ticks=none
]
\addplot graphics [includegraphics cmd=\pgfimage,xmin=-0.5, xmax=193.5, ymin=193.5, ymax=-0.5] {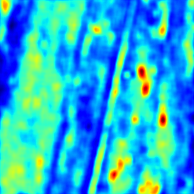};
\end{axis}

\end{tikzpicture}

%% file: figures/realdata/repeat-none/scaled-gaussian-kron-sgd.tex
\begin{tikzpicture}

\definecolor{darkgray176}{RGB}{176,176,176}

\begin{axis}[
colorbar,
colorbar style={ylabel={}},
colormap={mymap}{[1pt]
  rgb(0pt)=(0,0,0.5);
  rgb(22pt)=(0,0,1);
  rgb(25pt)=(0,0,1);
  rgb(68pt)=(0,0.86,1);
  rgb(70pt)=(0,0.9,0.967741935483871);
  rgb(75pt)=(0.0806451612903226,1,0.887096774193548);
  rgb(128pt)=(0.935483870967742,1,0.0322580645161291);
  rgb(130pt)=(0.967741935483871,0.962962962962963,0);
  rgb(132pt)=(1,0.925925925925926,0);
  rgb(178pt)=(1,0.0740740740740741,0);
  rgb(182pt)=(0.909090909090909,0,0);
  rgb(200pt)=(0.5,0,0)
},
height=4cm,
point meta max=54698.8538975652,
point meta min=24422.8415353895,
tick align=outside,
tick pos=left,
width=\textwidth,
x grid style={darkgray176},
xmin=-0.5, xmax=193.5,
y dir=reverse,
y grid style={darkgray176},
ymin=-0.5, ymax=193.5,
ticks=none
]
\addplot graphics [includegraphics cmd=\pgfimage,xmin=-0.5, xmax=193.5, ymin=193.5, ymax=-0.5] {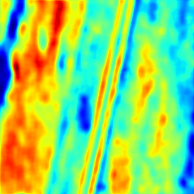};
\end{axis}

\end{tikzpicture}

%% file: figures/realdata/repeat-5/scaled-gaussian-glrt.tex
\begin{tikzpicture}

\definecolor{darkgray176}{RGB}{176,176,176}

\begin{axis}[
colorbar,
colorbar style={ylabel={}},
colormap={mymap}{[1pt]
  rgb(0pt)=(0,0,0.5);
  rgb(22pt)=(0,0,1);
  rgb(25pt)=(0,0,1);
  rgb(68pt)=(0,0.86,1);
  rgb(70pt)=(0,0.9,0.967741935483871);
  rgb(75pt)=(0.0806451612903226,1,0.887096774193548);
  rgb(128pt)=(0.935483870967742,1,0.0322580645161291);
  rgb(130pt)=(0.967741935483871,0.962962962962963,0);
  rgb(132pt)=(1,0.925925925925926,0);
  rgb(178pt)=(1,0.0740740740740741,0);
  rgb(182pt)=(0.909090909090909,0,0);
  rgb(200pt)=(0.5,0,0)
},
height=4cm,
point meta max=54698.8538975652,
point meta min=24422.8415353895,
tick align=outside,
tick pos=left,
width=\textwidth,
x grid style={darkgray176},
xmin=-0.5, xmax=193.5,
y dir=reverse,
y grid style={darkgray176},
ymin=-0.5, ymax=193.5,
ticks=none
]
\addplot graphics [includegraphics cmd=\pgfimage,xmin=-0.5, xmax=193.5, ymin=193.5, ymax=-0.5] {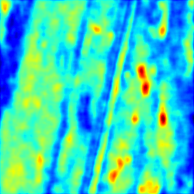};
\end{axis}

\end{tikzpicture}

%% file: figures/realdata/repeat-5/scaled-gaussian-sgd.tex
\begin{tikzpicture}

\definecolor{darkgray176}{RGB}{176,176,176}

\begin{axis}[
colorbar,
colorbar style={ylabel={}},
colormap={mymap}{[1pt]
  rgb(0pt)=(0,0,0.5);
  rgb(22pt)=(0,0,1);
  rgb(25pt)=(0,0,1);
  rgb(68pt)=(0,0.86,1);
  rgb(70pt)=(0,0.9,0.967741935483871);
  rgb(75pt)=(0.0806451612903226,1,0.887096774193548);
  rgb(128pt)=(0.935483870967742,1,0.0322580645161291);
  rgb(130pt)=(0.967741935483871,0.962962962962963,0);
  rgb(132pt)=(1,0.925925925925926,0);
  rgb(178pt)=(1,0.0740740740740741,0);
  rgb(182pt)=(0.909090909090909,0,0);
  rgb(200pt)=(0.5,0,0)
},
height=4cm,
point meta max=54698.8538975652,
point meta min=24422.8415353895,
tick align=outside,
tick pos=left,
width=\textwidth,
x grid style={darkgray176},
xmin=-0.5, xmax=193.5,
y dir=reverse,
y grid style={darkgray176},
ymin=-0.5, ymax=193.5,
ticks=none
]
\addplot graphics [includegraphics cmd=\pgfimage,xmin=-0.5, xmax=193.5, ymin=193.5, ymax=-0.5] {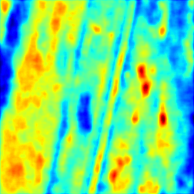};
\end{axis}

\end{tikzpicture}

%% file: figures/realdata/repeat-5/scaled-gaussian-kron-glrt.tex
\begin{tikzpicture}

\definecolor{darkgray176}{RGB}{176,176,176}

\begin{axis}[
colorbar,
colorbar style={ylabel={}},
colormap={mymap}{[1pt]
  rgb(0pt)=(0,0,0.5);
  rgb(22pt)=(0,0,1);
  rgb(25pt)=(0,0,1);
  rgb(68pt)=(0,0.86,1);
  rgb(70pt)=(0,0.9,0.967741935483871);
  rgb(75pt)=(0.0806451612903226,1,0.887096774193548);
  rgb(128pt)=(0.935483870967742,1,0.0322580645161291);
  rgb(130pt)=(0.967741935483871,0.962962962962963,0);
  rgb(132pt)=(1,0.925925925925926,0);
  rgb(178pt)=(1,0.0740740740740741,0);
  rgb(182pt)=(0.909090909090909,0,0);
  rgb(200pt)=(0.5,0,0)
},
height=4cm,
point meta max=54698.8538975652,
point meta min=24422.8415353895,
tick align=outside,
tick pos=left,
width=\textwidth,
x grid style={darkgray176},
xmin=-0.5, xmax=193.5,
y dir=reverse,
y grid style={darkgray176},
ymin=-0.5, ymax=193.5,
ticks=none
]
\addplot graphics [includegraphics cmd=\pgfimage,xmin=-0.5, xmax=193.5, ymin=193.5, ymax=-0.5] {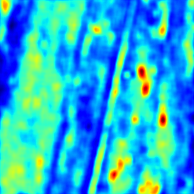};
\end{axis}

\end{tikzpicture}

%% file: figures/realdata/repeat-5/scaled-gaussian-kron-sgd.tex
\begin{tikzpicture}

\definecolor{darkgray176}{RGB}{176,176,176}

\begin{axis}[
colorbar,
colorbar style={ylabel={}},
colormap={mymap}{[1pt]
  rgb(0pt)=(0,0,0.5);
  rgb(22pt)=(0,0,1);
  rgb(25pt)=(0,0,1);
  rgb(68pt)=(0,0.86,1);
  rgb(70pt)=(0,0.9,0.967741935483871);
  rgb(75pt)=(0.0806451612903226,1,0.887096774193548);
  rgb(128pt)=(0.935483870967742,1,0.0322580645161291);
  rgb(130pt)=(0.967741935483871,0.962962962962963,0);
  rgb(132pt)=(1,0.925925925925926,0);
  rgb(178pt)=(1,0.0740740740740741,0);
  rgb(182pt)=(0.909090909090909,0,0);
  rgb(200pt)=(0.5,0,0)
},
height=4cm,
point meta max=54698.8538975652,
point meta min=24422.8415353895,
tick align=outside,
tick pos=left,
width=\textwidth,
x grid style={darkgray176},
xmin=-0.5, xmax=193.5,
y dir=reverse,
y grid style={darkgray176},
ymin=-0.5, ymax=193.5,
ticks=none
]
\addplot graphics [includegraphics cmd=\pgfimage,xmin=-0.5, xmax=193.5, ymin=193.5, ymax=-0.5] {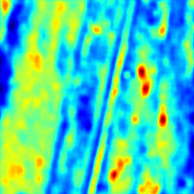};
\end{axis}

\end{tikzpicture}

%% file: figures/realdata/repeat-10/scaled-gaussian-glrt.tex
\begin{tikzpicture}

\definecolor{darkgray176}{RGB}{176,176,176}

\begin{axis}[
colorbar,
colorbar style={ylabel={}},
colormap={mymap}{[1pt]
  rgb(0pt)=(0,0,0.5);
  rgb(22pt)=(0,0,1);
  rgb(25pt)=(0,0,1);
  rgb(68pt)=(0,0.86,1);
  rgb(70pt)=(0,0.9,0.967741935483871);
  rgb(75pt)=(0.0806451612903226,1,0.887096774193548);
  rgb(128pt)=(0.935483870967742,1,0.0322580645161291);
  rgb(130pt)=(0.967741935483871,0.962962962962963,0);
  rgb(132pt)=(1,0.925925925925926,0);
  rgb(178pt)=(1,0.0740740740740741,0);
  rgb(182pt)=(0.909090909090909,0,0);
  rgb(200pt)=(0.5,0,0)
},
height=4cm,
point meta max=54698.8538975652,
point meta min=24422.8415353895,
tick align=outside,
tick pos=left,
width=\textwidth,
x grid style={darkgray176},
xmin=-0.5, xmax=193.5,
y dir=reverse,
y grid style={darkgray176},
ymin=-0.5, ymax=193.5,
ticks=none
]
\addplot graphics [includegraphics cmd=\pgfimage,xmin=-0.5, xmax=193.5, ymin=193.5, ymax=-0.5] {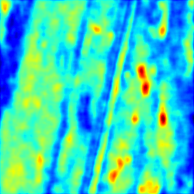};
\end{axis}

\end{tikzpicture}

%% file: figures/realdata/repeat-10/scaled-gaussian-sgd.tex
\begin{tikzpicture}

\definecolor{darkgray176}{RGB}{176,176,176}

\begin{axis}[
colorbar,
colorbar style={ylabel={}},
colormap={mymap}{[1pt]
  rgb(0pt)=(0,0,0.5);
  rgb(22pt)=(0,0,1);
  rgb(25pt)=(0,0,1);
  rgb(68pt)=(0,0.86,1);
  rgb(70pt)=(0,0.9,0.967741935483871);
  rgb(75pt)=(0.0806451612903226,1,0.887096774193548);
  rgb(128pt)=(0.935483870967742,1,0.0322580645161291);
  rgb(130pt)=(0.967741935483871,0.962962962962963,0);
  rgb(132pt)=(1,0.925925925925926,0);
  rgb(178pt)=(1,0.0740740740740741,0);
  rgb(182pt)=(0.909090909090909,0,0);
  rgb(200pt)=(0.5,0,0)
},
height=4cm,
point meta max=54698.8538975652,
point meta min=24422.8415353895,
tick align=outside,
tick pos=left,
width=\textwidth,
x grid style={darkgray176},
xmin=-0.5, xmax=193.5,
y dir=reverse,
y grid style={darkgray176},
ymin=-0.5, ymax=193.5,
ticks=none
]
\addplot graphics [includegraphics cmd=\pgfimage,xmin=-0.5, xmax=193.5, ymin=193.5, ymax=-0.5] {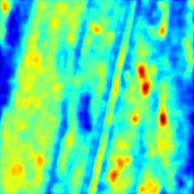};
\end{axis}

\end{tikzpicture}

%% file: figures/realdata/repeat-10/scaled-gaussian-kron-glrt.tex
\begin{tikzpicture}

\definecolor{darkgray176}{RGB}{176,176,176}

\begin{axis}[
colorbar,
colorbar style={ylabel={}},
colormap={mymap}{[1pt]
  rgb(0pt)=(0,0,0.5);
  rgb(22pt)=(0,0,1);
  rgb(25pt)=(0,0,1);
  rgb(68pt)=(0,0.86,1);
  rgb(70pt)=(0,0.9,0.967741935483871);
  rgb(75pt)=(0.0806451612903226,1,0.887096774193548);
  rgb(128pt)=(0.935483870967742,1,0.0322580645161291);
  rgb(130pt)=(0.967741935483871,0.962962962962963,0);
  rgb(132pt)=(1,0.925925925925926,0);
  rgb(178pt)=(1,0.0740740740740741,0);
  rgb(182pt)=(0.909090909090909,0,0);
  rgb(200pt)=(0.5,0,0)
},
height=4cm,
point meta max=54698.8538975652,
point meta min=24422.8415353895,
tick align=outside,
tick pos=left,
width=\textwidth,
x grid style={darkgray176},
xmin=-0.5, xmax=193.5,
y dir=reverse,
y grid style={darkgray176},
ymin=-0.5, ymax=193.5,
ticks=none
]
\addplot graphics [includegraphics cmd=\pgfimage,xmin=-0.5, xmax=193.5, ymin=193.5, ymax=-0.5] {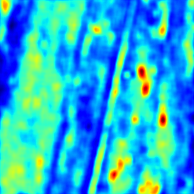};
\end{axis}

\end{tikzpicture}

%% file: figures/realdata/repeat-10/scaled-gaussian-kron-sgd.tex
\begin{tikzpicture}

\definecolor{darkgray176}{RGB}{176,176,176}

\begin{axis}[
colorbar,
colorbar style={ylabel={}},
colormap={mymap}{[1pt]
  rgb(0pt)=(0,0,0.5);
  rgb(22pt)=(0,0,1);
  rgb(25pt)=(0,0,1);
  rgb(68pt)=(0,0.86,1);
  rgb(70pt)=(0,0.9,0.967741935483871);
  rgb(75pt)=(0.0806451612903226,1,0.887096774193548);
  rgb(128pt)=(0.935483870967742,1,0.0322580645161291);
  rgb(130pt)=(0.967741935483871,0.962962962962963,0);
  rgb(132pt)=(1,0.925925925925926,0);
  rgb(178pt)=(1,0.0740740740740741,0);
  rgb(182pt)=(0.909090909090909,0,0);
  rgb(200pt)=(0.5,0,0)
},
height=4cm,
point meta max=54698.8538975652,
point meta min=24422.8415353895,
tick align=outside,
tick pos=left,
width=\textwidth,
x grid style={darkgray176},
xmin=-0.5, xmax=193.5,
y dir=reverse,
y grid style={darkgray176},
ymin=-0.5, ymax=193.5,
ticks=none
]
\addplot graphics [includegraphics cmd=\pgfimage,xmin=-0.5, xmax=193.5, ymin=193.5, ymax=-0.5] {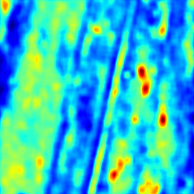};
\end{axis}

\end{tikzpicture}